\documentclass{article}
\usepackage[utf8]{inputenc}

\usepackage{amsmath,amsfonts,amsthm,amssymb,amsxtra,bm}
\usepackage{setspace}
\usepackage{Tabbing}
\usepackage{fancyhdr}
\usepackage{lastpage}
\usepackage{authblk}
\usepackage{extramarks}
\usepackage{chngpage}
\usepackage{soul,color}
\usepackage{graphicx,wrapfig}
\usepackage{tabularx}
\usepackage{caption}
\usepackage{bbm}
\usepackage[dvipsnames]{xcolor}
\usepackage{hyperref}
\hypersetup{
    colorlinks=true,
    linkcolor={red!50!black},
    citecolor={blue!50!black},
    urlcolor={blue!80!black}
}
\usepackage{cleveref}
\usepackage{apacite}
\usepackage{dirtytalk}

\usepackage{floatrow}
\usepackage{subfig}
\usepackage{enumitem}

\newtheorem{theorem}{Theorem}[section]
\newtheorem{lemma}[theorem]{Lemma}
\newtheorem{proposition}[theorem]{Proposition}
\newtheorem{corollary}[theorem]{Corollary}

\newtheorem{assumption}[theorem]{Assumption}

\theoremstyle{remark}

\newcommand{\R}{\mathbb{R}}

\newcommand{\E}[2][{}]{\mathbb{E}_{#1}\hspace{-0.15em}\left[#2\right]}

\newcommand{\td}[2]{\frac{d#1}{d#2}}

\newcommand{\sumdu}[2]{\overset{#2}{\underset{#1}{\sum}}\,}

\newcommand{\limd}[1]{\underset{#1}{\lim }\,}

\newcommand{\remove}[1]{}

\title{Feedback dynamics in matching networks drive behavioral differentiation despite overlapping objectives\footnote{An earlier version of the mathematical theory presented here was developed in collaboration with Athanasios Kehagias (Aristotle University of Thessaloniki).}}
\author{Alexandros Gelastopoulos}

\affil{Department of Social and Behavioral Sciences,\\ Toulouse School of Economics,\\
University of Toulouse Capitole,\\
Toulouse 31080, France.
}

\affil{Department of Business and Management,\\
University of Southern Denmark,\\
Odense DK-5230, Denmark.
}

\affil{Email: alexandros.gelastopoulos@iast.fr\\
ORCID: 0000-0003-1428-0130}
\begin{document}

\maketitle

\begin{abstract}
    Many bipartite social networks exhibit pronounced asymmetries in selectivity and matching opportunities: members of one side can afford to be highly selective, while members of the opposite side are forced to accept less desirable matches. While it is natural to try to explain this asymmetry in terms of the intrinsic characteristics of the two sides or other exogenous factors, here we show that such asymmetries can also emerge endogenously through a feedback process generated by the matching process itself: as one side becomes more selective, the other side is pushed to be less selective due to reduced matching opportunities, and vice versa. We develop a model in which individuals repeatedly form one-to-one matches across two groups and adapt their selectivity to achieve a target matching rate. Using both analytic and numerical methods, we show that when encounters are sufficiently frequent, the unique equilibrium is for one group to be highly selective and the other non-selective. This qualitative outcome holds even for heterogeneous groups with overlapping, almost indistinguishable distributions of target matching rates. The model makes several testable predictions, and it provides a mechanism for behavioral differentiation in repeated matching environments, with applications ranging from online dating to hiring and housing markets.
\end{abstract}

\section{Introduction}

Many social interactions can be modeled as matches that form between members of two distinct groups of agents: employers and people looking for jobs, landlords and apartment seekers, men and women looking for partners, scientific conferences and researchers who submit their work, and so on. Often these bipartite networks are highly skewed, in the sense that agents on one side have plenty of opportunities to match and, as a result, can afford to be highly selective, while agents on the opposite side have to put a lot of effort into finding any match. While this is unsurprising when there is scarcity of agents of one type and abundance of the other, this high imbalance in selectivity and matching opportunities can occur even when the two pools of agents have roughly equal numbers and matches are one-to-one. 

A striking example of this is matching on dating apps: in a controlled experiment with fake profiles on Tinder, \citeA{tyson2016first} found that female profiles receive 100 times more heterosexual `likes' than male profiles do. This extreme imbalance cannot be explained by the gender ratio of dating app users alone, which empirical estimates suggest is only moderately male-skewed \cite{SSRS2025OnlineDating}. Although it is tempting to try to explain such asymmetries in terms of sex differences in dating goals, either evolutionarily established or socioculturally driven, here we ask whether there is something inherent in the dynamics of matching networks that drives one side to be selective and the other non-selective.
Answering this will give insight into other types of bipartite matching networks, including hiring networks, the housing market, and academic paper submissions.

The core idea of our paper is the following: because every match in such networks involves one member from each of the two groups, the more selective one of the two groups becomes, the less selective the opposite group can afford to be and still find a match, and vice versa. This creates a feedback loop that can potentially exaggerate any pre-existing imbalances in population sizes or willingnesses to match. Still, it is unclear whether this mechanism can trigger a runaway process that leads the `in demand' group to become orders of magnitude more selective than the opposite group or its effect is only moderate. By formalizing the matching dynamics as a mathematical model where agents adjust their selectivity to reach a desired matching rate, we show using both analytic techniques and numerical simulations that even tiny differences either in size or in mean target matching rates between the two groups can lead to near-diametrically opposite behavior in terms of selectivity. This result depends on the rate of encounters being sufficiently large (so that it makes sense for anyone to be selective), but it is highly robust to parameter choices otherwise.

The bipartite networks that we are modeling have the following properties: First, we are interested in cases where individuals form matches repeatedly (e.g., casual heterosexual relations) rather than exiting the market, permanently or temporarily, when they find a match (e.g., marriage market). However, we assume that individuals have a soft limit on the rate that they can form matches, in the sense that if one forms matches too frequently for their personal taste, they start becoming more selective. We model this as a \textit{target matching rate} which the individual tries to achieve by adjusting their selectivity. This is an essential feature of our model, as we are interested in what drives one side to be more selective than the other.

Additionally, we restrict our attention to matching with \textit{non-transferable utility}, i.e., when there is no form of currency (price, salary, rent, etc.) transferred between two matched individuals, which could be adjusted to make a match more attractive for one side. Finally, agents are allowed to have idiosyncratic preferences, that is, we do not assume that there is a single ranking of the individuals of one group upon which the opposite group agrees uniformly. Although in our base model preferences are entirely uncorrelated, we do consider an extension where some individuals are considered on average more attractive than others.

To summarize, we model situations with the following characteristics:
\begin{enumerate}
    \item\label{distinctGroups} There are two non-overlapping groups of agents who try to match across groups.
    \item\label{NTU} Utility is non-transferable.
    \item\label{repeatedMatches} Agents form matches repeatedly.
    \item\label{softLimits} Agents have soft limits on their matching rates.
    \item\label{idiosyncratic}Preferences are idiosyncratic.
\end{enumerate}
Examples satisfying these conditions include casual heterosexual relations, hiring interviews (where a match is defined as the candidate being \textit{interviewed} for a position), the housing market (where a match is an apartment visit), academic conferences (accepted submissions), and various paid services where the remuneration rate is fixed.

Our results can be summarized as follows: using a model that captures properties \labelcref{distinctGroups,NTU,repeatedMatches,softLimits,idiosyncratic}, we show the existence of a unique equilibrium which is globally attractive, i.e., no matter the initial conditions, each agent's selectivity converges to a fixed (agent-dependent) value. We also obtain an analytic expression for the equilibrium when agents of each group are homogeneous with respect to how often they would like to match (equal target matching rates), and a semi-analytic expression in the general case. In the former, special case, the result immediately shows that even tiny differences in target matching rates can drive the selectivity of the two groups to extreme opposites, even for equal group sizes. In the general case, we illustrate the same phenomenon through simulations: tiny differences in the \textit{average} target matching rates of the two groups can lead to near-diametrically opposite selectivities even with equal group sizes. Because convergence to a unique equilibrium is theoretically guaranteed even in the general case, the simulation results do not depend on the choice of initial conditions, which makes it safe to draw general conclusions. Our model predicts higher behavioral asymmetries for high rates of encounters, smaller asymmetries between \textit{recent} members of the two groups, and insensitivity to group size ratio up to a critical value with threshold behavior at that value.

Our results also reveal a surprising consequence of the dynamics in such networks: if the distributions of target matching rates of the two groups are overlapping, then there will be individuals in the in-demand group who are more promiscuous (have higher target matching rates) than many individuals in the opposite group. Yet even these promiscuous individuals of the in-demand group will end up being highly selective, while less promiscuous individuals of the opposite group will be much less selective. In other words, group membership is a stronger determinant of behavior than individual goals are. More generally, our results suggest that feedback mechanisms that are structurally embedded in matching networks can create large asymmetries in behavior which do not necessarily reflect equally large differences in goals.

\subsection{Relation to the matching literature in economics}

The study of matching has a long history in economics. Stable matching theory, as established by \citeA{gale1962college} (see also \citeNP{roth1992two}), deals with the question of what matching assignments have the property that no two individuals would prefer matching with one another instead of their assigned partners. It takes the point of view of a central planner who knows the preferences of every individual and needs to algorithmically find a match that satisfies the above property. Although the assumption of a central planner might be realistic in some cases, here we are interested in distributed markets, where each individual needs to go through the available options for themselves.\footnote{It is worth noting an interesting parallel to our main result in stable matching theory: \citeA{ashlagi2017unbalanced} show that even a single extra member on one side of the market can dramatically alter stable outcomes so that, in every stable matching, almost every member of the smaller group is matched with one of their top choices, while the larger group gets partners that are almost random in rank. While the setting and mechanics are very different, this is another illustration that matching markets can amplify small asymmetries.}

A more closely related literature is that of search-theoretic matching. Here there is no central planner, but people rather encounter others randomly over time and decide whether they accept them as partners (for a review see \citeNP{chade2017sorting}; for the non-transferable utility case specifically see also \citeNP{lauermann2025matching}). Within this literature, a distinguishing feature of our model is that people keep track of the rate that they have been forming matches, so that they may adjust their selectivity accordingly. In contrast, existing literature assumes agents to be `Markovian', with the past having no effect on their expected (future) utility function and consequently on their decisions.

Allowing for idiosyncratic preferences is also a distinguishing feature of our model. Arguably, most social systems involve individuals with heterogeneous preferences, yet the existing literature has largely failed to incorporate such heterogeneity. As \citeNP[pp. 517-518]{chade2017sorting} put it, ``the literature frontier assumes a common evaluation of agents, without a hint that beauty is in the eye of the beholder'', and the situation hasn't changed until today.\footnote{\citeA{adachi2003search,lauermann2014stable} do consider the case of idiosyncratic preferences and provide some results in the limit of vanishing search frictions, i.e., high rates of encounters. These are the only cases with idiosyncratic preferences reviewed by \citeA[Ch.~6]{lauermann2025matching}.}

Another difference is that standard treatments of search-theoretic matching deal with the case of fully rational individuals who have full knowledge of everyone's preferences. Our model, in contrast, has an interpretation of each individual being aware only of their own experience and adjusting their behavior over time based on this experience and their goals. In other words, our agents are boundedly rational and adjust their selectivity adaptively.

Finally, the economics literature specifies agent preferences which allows to characterize equilibria, but it typically does not specify any dynamics when the system is out of equilibrium. Here we take a dynamical systems approach, describing the evolution of the dynamic variables via differential equations, which allows us to study also whether equilibria are reached from arbitrary initial conditions. We prove our results under minimal assumptions on the functional form of these differential equations (subject to properties \labelcref{distinctGroups,NTU,repeatedMatches,softLimits,idiosyncratic} of the previous section), which makes them robust to specific modeling choices.

\section{Main results}

\subsection{Model description}
\label{sectionModel}

We consider two populations of individuals, which we will call male and female, with sizes $M$ and $N$ respectively. Time is continuous and people meet others of the opposite type randomly at a certain rate. Whenever two individuals meet, each decides whether they accept the other as a partner, and if they both accept, a match is formed. Matches do not last, so that paired individuals immediately return to the pool of available individuals. This assumption makes sure that the numbers of available individuals remain constant.

Every individual has a personal, constant \textit{target matching rate} denoted by $c_i>0$ and $d_j>0$ for the $i$-th male and $j$-th female, respectively. This is how often the individual would ideally form matches, assuming there were enough available partners. To achieve this ideal rate, individuals may adjust how selective they are. We model their selectivity as an acceptance probability, i.e., with what probability they accept a person they encounter as a partner, denoted by $a_i(t)$ and $b_j(t)$ for the $i$-th male and $j$-th female, respectively. In the base model, preferences are uncorrelated, i.e., beauty is entirely in the eye of the beholder, so that $a_i(t)$ characterizes the probability that the $i$-th male accepts \textit{any} female, and similarly for $b_j(t)$.

We assume that the acceptance probabilities $a_i(t)$ and $b_j(t)$ satisfy the differential equations

\begin{equation}
\label{mainEq1a}
    \td{a_i(t)}{t}=r \cdot \left[c_i-K\cdot a_i(t)\cdot \sum_jb_j(t)\right], \ \ \text{ subject to }a_i(t)\in [0,1],
\end{equation}
and
\begin{equation}
\label{mainEq1b}
    \td{b_j(t)}{t}=r \cdot \left[d_j-K\cdot b_j(t)\cdot \sum_ia_i(t)\right], \ \ \text{ subject to }b_j(t)\in [0,1],
\end{equation}
where $r$ is a positive constant. Unless otherwise specified, summation over $i$ will always be from $1$ to $M$, and summation over $j$ will be from $1$ to $N$.

The interpretation of these equations is the following: The probability that the $i$-th male individual matches with the $j$-th female in an interval of time $dt$ is equal to the probability they meet ($K\cdot dt$) times the probability they both accept each other ($a_i(t)\cdot b_j(t)$), that is, $K\cdot a_i(t)\cdot b_j(t)\cdot dt$. Hence $K\cdot a_i(t)\cdot b_j(t)$ is the instantaneous rate at which these two individuals match. Consequently, the term $K\cdot a_i(t)\cdot \sumdu{j}{}b_j(t)$ in \cref{mainEq1a} gives the instantaneous rate that the $i$-th male matches with \textit{any} female. If this rate is smaller than their target matching rate $c_i$, then $\td{a_i(t)}{t}$ is positive, so the acceptance probability $a_i(t)$ increases, or equivalently, the individual becomes less selective. Conversely, if the experienced matching rate ($K\cdot a_i(t)\cdot \sumdu{j}{}b_j(t)$) is larger than the target rate $c_i$, then $a_i(t)$ decreases, meaning that the individual becomes more selective.\footnote{Note that if we wanted to give a precise definition of selectivity, we could define it as $1-a_i(t)$ (or $1-b_j(t)$), i.e., the probability of \textit{rejecting} a potential partner. All our formal results, however, are stated directly with respect to the acceptance probabilities $a_i(t)$ and $b_j(t)$.} The parameter $r$ controls how fast this adjustment happens. The justification is entirely analogous for \cref{mainEq1b}. The restrictions $a_i(t)\in [0,1]$ and $b_j(t)\in [0,1]$ are needed because these variables represent probabilities. The precise way we handle the dynamics at the boundaries is detailed in \cref{sectionTechnical}, where all our technical results are gathered.

In \cref{mainEq1a,mainEq1b}, the rate that individuals update their selectivity is proportional to the distance of their actual matching rate from their target rate. We use this form here for simplicity, but our results hold under much less restrictive assumptions. In fact, barring some technical details, the only requirement is simply that individuals update their selectivity \textit{in the direction} of achieving their target matching rate (see \cref{sectionPreliminaries}). No assumptions are required regarding the rate of updating. Having a concrete form like the system of \cref{mainEq1a,mainEq1b} allows us to run simulations, but all our theoretical results are proved under the less restrictive condition.

\subsection{Analytical results}
\label{sectionAnalyticalResults}

The system of \cref{mainEq1a,mainEq1b} is an agent-based model, with each agent's behavior represented by a single variable ($a_i(t)$ or $b_j(t)$) that determines how selective they are. The question we ask is how selective they become in the long term. Specifically, does the system reach an equilibrium with each agent exhibiting a constant (but agent-dependent) selectivity, and if so, how selective are they? Once we answer this question, we can compare the distributions of selectivities in the two groups.

Our main result says that, barring a borderline case where the target matching rates on the male and female sides exactly balance out, the system has a unique equilibrium, which is reached from any initial conditions. Here and in what follows, we are using bold-face notation to denote mathematical vectors.

\begin{theorem}[Existence and uniqueness of equilibrium]
\label{theoremMain}
    If $\sum_ic_i\neq \sum_jd_j$, then the system of \cref{mainEq1a,mainEq1b} has a unique equilibrium $({\bf\hat a},{\bf\hat b})=(\hat a_1,\ldots ,\hat a_M,\hat b_1,\ldots \hat b_N)$, which is globally attractive.
\end{theorem}

This theorem is saying that, barring the special case $\sum_ic_i=\sum_jd_j$, no matter the initial conditions, each agent's selectivity will converge over time to a certain value, which can be obtained from the parameters of the system (numbers of individuals, rate of encounters, and target matching rates). This means that, in order to understand what happens in the long run, we need only find the equilibrium values, that is, $\hat a_1,\ldots, \hat a_M,\hat b_1,\ldots ,\hat b_N$. The uniqueness of equilibrium fails in the case that the target matching rates on the two sides exactly balance out ($\sum_ic_i=\sum_jd_j$). Although it is still possible to characterize the long-term behavior in this case, we do not pursue this here, as this is a corner case that is unlikely to be exactly satisfied in any real system.

Returning to the generic case ($\sum_ic_i\neq \sum_jd_j$), the next result specifies the selectivities at equilibrium for homogeneous populations, i.e., populations where all males have equal target matching rates $c$ and all females have equal target matching rates $d$. Note that in that case, the assumption $\sum_ic_i\neq \sum_jd_j$ becomes $Mc\neq Nd$. Without loss of generality (by relabeling male vs. female if necessary), we may assume that $Mc>Nd$.

\begin{proposition}[Equilibrium values, homogeneous groups]
\label{propositionHomogeneousEq}
    Consider the system of \cref{mainEq1a,mainEq1b} and suppose that $c_i=c$ for all $i$, and $d_j=d$ for all $j$, where $c,d$ are some positive constants. Without loss of generality, assume that $cM>dN$. Then,
        \begin{itemize}
            \item $\hat a_i=1$ for all $i$.
            \item $\hat b_j=\min\{1,\frac{d}{KM}\}$, for all $j$.
        \end{itemize}
\end{proposition}

The implications of this result are striking. Consider for simplicity that $M=N$. Then, the proposition says that even with a tiny difference in the target matching rates of males and females $c$ and $d$, all males will necessarily be driven to become non-selective ($a_i(t)\to \hat a_i=1$), i.e., `accept' any female as a possible match. In contrast, females will tend to accept males at rate $\frac{d}{KM}$, which is possibly very small. For example, if a female meets $KM=100$ males per unit of time and her target matching rate is $d=1$, then $\hat b_j=\frac{d}{KM}=0.01$. See \cref{fig:homogeneous}. In the case $d>KM$, i.e., if the target matching rate of females exceeds the rate they meet males, then females also become non-selective.

\begin{figure}[ht]
    \centering
    \includegraphics[width=\linewidth]{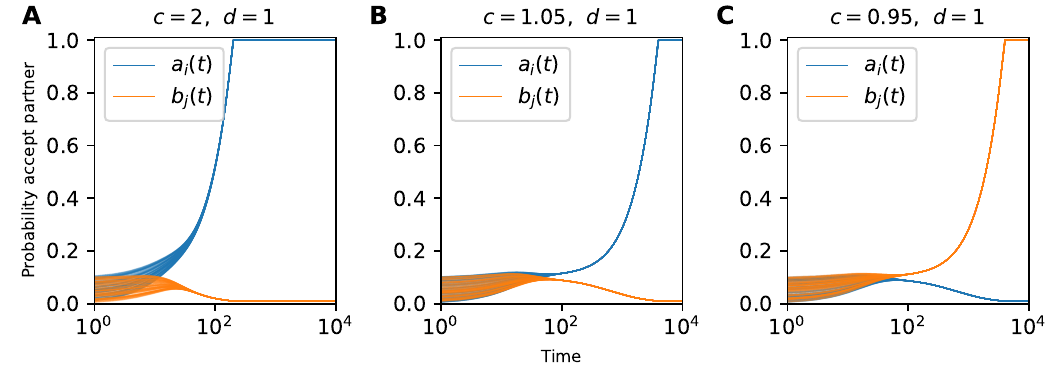}
    \caption{Simulation of the system of \cref{mainEq1a,mainEq1b} with homogeneous within-group target matching rates \textbf{(A)} males: $c=2$, females: $d=1$, \textbf{(B)} males: $c=1.05$, females: $d=1$, \textbf{(C)} males: $c=0.95$, females: $d=1$. In (A) and (B), all males are driven to eventually be non-selective (acceptance rate $a_i(t)\to 1$) and males highly selective ($b_j(t)\approx 0$). In (B), this is the case despite the highly similar target matching rates of males and females. The exact opposite behavior is observed in (C), where the target matching rates are still similar in size, but now smaller for males. Other parameters are $M=N=100$, $r=0.005$, $K=1$, and $a_i(0)$ and $b_j(0)$ are drawn uniformly randomly from the interval $(0, 0.1)$.}
    \label{fig:homogeneous}
\end{figure}

For heterogeneous populations, we can still find the equilibrium values semi-analytically, but the statement of the result is more involved. We present it here for completeness, but it may be skipped on a first reading (continue to \cref{sectionSimulations}).

To express the result in the general case, it is easiest if we assume, without loss of generality, that $c_1\leq \ldots \leq c_M$ and $d_1\leq\ldots\leq d_N$. We describe the equilibrium in two steps. First, we give an expression for the numbers of males and females whose acceptance probabilities at equilibrium are not saturated, i.e., they remain selective to at least some degree, denoted by $\hat i$ and $\hat j$. That is, we define
\begin{equation}
\label{defIstarJstar}
    \hat i=|\{i:\hat a_i<1\}|,\ \ \ \hat j=|\{j:\hat b_j<1\}|.
\end{equation}
The following proposition allows us to get these values directly from the parameters of the system.
\begin{proposition}
\label{propositionIstarJstarExpressions}
Assume $c_1\leq \ldots \leq c_M$, $d_1\leq\ldots\leq d_N$, and $\sum_ic_i\neq \sum_jd_j$, and define $\hat i$ and $\hat j$ through \cref{defIstarJstar}. We have
\begin{equation}
\label{expressionIstar}
    \hat i=\max\,\left\{i:\sum_j\min\,\left\{\frac{d_j}{(M-i)c_i+\sumdu{k=1}{i}c_{k}},\frac{K}{c_i}\right\}>1\right\}
\end{equation}
and
\begin{equation}
\label{expressionJstar}
    \hat j=\max\,\left\{j:\sum_i\min\,\left\{\frac{c_i}{(N-j)d_j+\sumdu{k=1}{j}d_{k}},\frac{K}{d_j}\right\}>1\right\},
\end{equation}
with the convention that $\max\,\emptyset=0$.
\end{proposition}
The outer sum in \cref{expressionIstar} is decreasing in $c_i$, so efficient algorithms can be used to find the maximum index for which this sum is greater than $1$, and similarly for \cref{expressionJstar}.

Once we have $\hat i$ and $\hat j$, the equilibrium can be found as follows:

\begin{proposition}[Equilibrium values, general]
\label{propositionHeterogeneousEq}
    Consider the system of \cref{mainEq1a,mainEq1b} and suppose that $c_1\leq \ldots \leq c_M$, $d_1\leq\ldots\leq d_N$, and $\sum_ic_i\neq \sum_jd_j$. Define $\hat i$ and $\hat j$ through \cref{defIstarJstar} and further denote
    \begin{align}
    \label{defCDL}
        C^* =\sum_{i\leq \hat i}c_i, \hspace{2em} D^*=\sum_{j\leq \hat j}d_j,\hspace{2em} L =(D^*-C^*)-K(M-\hat i)(N-\hat j).
    \end{align}
    Then, the unique equilibrium $({\bf \hat a},{\bf \hat b})$ satisfies
    \begin{equation}
    \label{eq43}
        \hat b_j=\min\,\left\{\frac{d_j}{Kx^*},1\right\},\ \text{ and } \ \hat a_{i}=\min\,\left\{\frac{c_i}{K\sum_j\hat b_j},1\right\},
    \end{equation}
    where
    \begin{equation}
    \begin{aligned}
        x^*=\left\{\begin{array}{cc}
            \frac{-L+\sqrt{L^2+4KD^*(M-\hat i)(N-\hat j)}}{2K(N-\hat j)}, & \text{ if }\hat j<N\\ \\
            \frac{(M-\hat i)D^*}{D^*-C^*}, & \text{ if }\hat j=N.
        \end{array}
        \right.
    \end{aligned}
    \end{equation}
\end{proposition}

\subsection{Simulations}
\label{sectionSimulations}

\Cref{propositionHomogeneousEq} showed that for homogeneous populations, even minuscule differences in the total target matching rate of the two groups ($Mc$ vs. $Nd$) can result in one side becoming highly selective while the other becomes entirely non-selective. In contrast, for heterogeneous populations, although \cref{propositionHeterogeneousEq} gives the equilibrium values, it is not immediately informative regarding whether small differences in the distributions of the target matching rates of the two groups can lead to large differences in the selectivity distributions. We examine this question using numerical simulations.

\Cref{fig:heterogeneous} shows a simulation with 100 males and 100 females ($M=N=100$), and with the target matching rates $c_i$ and $d_j$ drawn from uniform distributions. Specifically we draw the $c_i$ i.i.d. from a uniform distribution on the interval $[0,2.5]$ and the $d_j$ i.i.d. from a uniform distribution on $[0,2]$. This results in the $c_i$'s having a slightly larger mean than the $d_j$'s ($\bar c:=\frac{1}{M}\cdot \sum_ic_i=1.25$ vs. $\bar d:=\frac{1}{N}\cdot \sum_jd_j=1$), but with the two distributions highly overlapping (\cref{fig:heterogeneous}A).

\begin{figure}
    \centering
    \includegraphics[width=\linewidth]{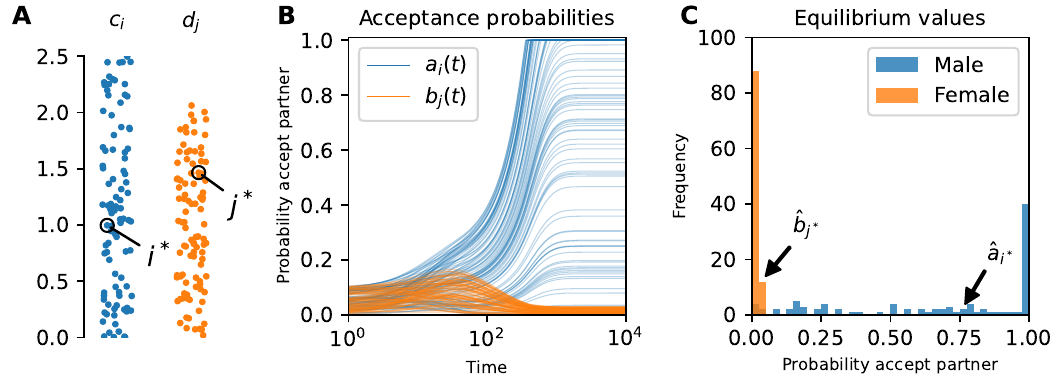}
    \caption{\textbf{Simulation of the system of \cref{mainEq1a,mainEq1b} for heterogeneous populations.} \textbf{(A)} We draw the $c_i$'s i.i.d. from a uniform distribution on $(0,2.5)$ and the $d_j$'s i.i.d. from a uniform distribution on $(0,2)$, resulting in $\E{c_i}=1.25$ and $\E{d_j}=1$. Other parameter values are as in \cref{fig:homogeneous}. \textbf{(B)} Trajectories of acceptance probabilities $a_i(t)$ and $b_j(t)$. \textbf{(C)} Histogram of equilibrium values $\hat a_i$ and $\hat b_j$. Despite the highly overlapping target matching rate distributions, selectivity at equilibrium is highly polarized between the two groups. Even for a specific pair of individuals $i^*$, $j^*$, for which the female has higher target matching rate than the male ($c_{i^*}<d_{j^*}$; circles in A), the female ends up becoming much more selective ($\hat a_{i^*}>\hat b_{j^*}$; arrows in C).}
    \label{fig:heterogeneous}
\end{figure}
As \cref{fig:heterogeneous} illustrates, despite the overlapping target matching rate distributions (\cref{fig:heterogeneous}A), the selectivity distributions at equilibrium are highly distinct (\cref{fig:heterogeneous}C). Specifically, all females are highly selective at equilibrium ($\hat b_j\leq 0.05$), while males exhibit a large variation in their selectivity, with about $40\%$ of them being non-selective ($\hat a_i=1$) and only a tiny fraction exhibiting selectivity comparable to the female population ($\hat a_i\leq 0.05$). The mean acceptance probabilities at equilibrium are $\frac{1}{M}\cdot \sum_i\hat a_i\approx 0.712$ and $\frac{1}{N}\cdot \sum_j\hat b_j\approx 0.014$.

\subsubsection{Group membership outweighs individual goals}

The fact that overlapping distributions of target matching rates can give rise to polarized distributions of equilibrium selectivities has an important consequence: if we pick a specific pair of individuals, a male $i^*$ and a female $j^*$, even if $j^*$ has a higher target matching rate ($d_{j^*}>c_{i^*}$; see circles in \cref{fig:heterogeneous}A), with high probability $j^*$ will nevertheless be more selective at equilibrium ($\hat b_{j^*}<\hat a_{i^*}$; arrows in \cref{fig:heterogeneous}C). This would be entirely unexpected if one considered only individual target matching rates. Without taking into account the feedback dynamics that lead to the highly polarized behavior, one would expect that a more `promiscuous' individual would be less selective (or at least only slightly more selective given the small differences in mean target matching rates of the two groups). But in fact, even promiscuous individuals of the in-demand group end up being highly selective, while the opposite is true for the other group. In other words, group membership determines behavior to a much greater extent than individual goals do.

\subsubsection{Threshold behavior of selectivity polarization}

If we make the distributions of the target matching rates $c_i$ and $d_j$ more similar to one another, one would expect the distributions of acceptance probabilities at equilibrium ($\hat a_i$ and $\hat b_j$) to also become more similar. However, making the distributions of $c_i$ and $d_j$ more similar has a virtually negligible effect on the equilibrium acceptance probabilities, as long as $\bar c>\bar d$: the large differences in acceptance probability persist even for target matching rates that differ on average by only a few percentage points (\cref{fig:varyingMeanC}A-C). Qualitative changes appear only when $\bar c\approx \bar d$ (\cref{fig:varyingMeanC}D-E) and the polarity reverses as soon as $\bar c<\bar d$ (\cref{fig:varyingMeanC}F).\footnote{Note that when we have exactly $\bar c=\bar d$, the assumptions of \cref{theoremMain} fail and we no more have a unique equilibrium. Although we do not show it here, the trajectories do converge, but the outcome depends on the initial conditions.} \Cref{fig:varyParams}A shows the mean equilibrium acceptance probablities (mean $\hat a_i$ and mean $\hat b_j$) as a function of $\bar c$.

\begin{figure}
    \centering
    \includegraphics[width=\linewidth]{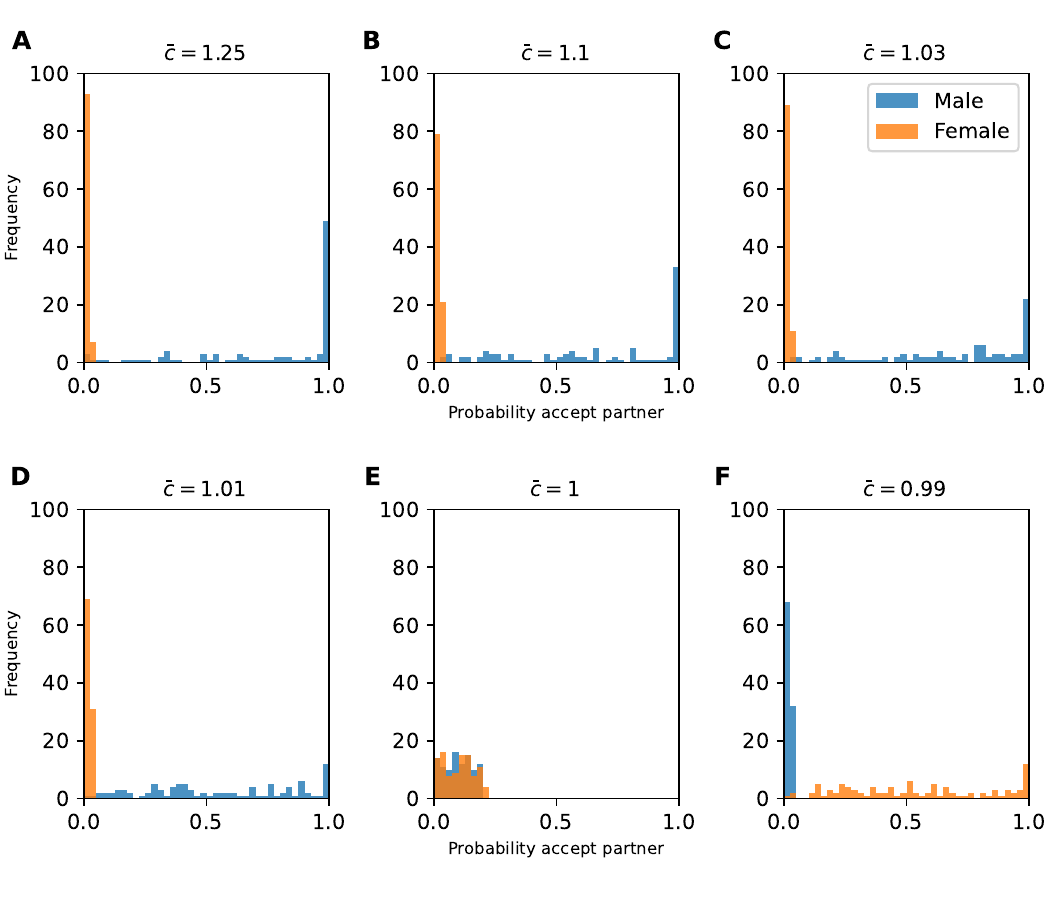}
    \caption{\textbf{Simulation of the system of \cref{mainEq1a,mainEq1b} with different mean target matching rates for males.} \textbf{(A-F)} The $c_i$'s are drawn i.i.d. from a uniform distribution on $(0,2\bar c)$, with $\bar c$ indicated at the top of each panel. In all cases, $\bar d=1$. All other parameters are as in \cref{fig:homogeneous}. The equilibrium selectivity distributions change very little until $\bar c\approx \bar d=1$, and they reverse when $\bar c$ falls below that value. The parameter values in (A) are identical to those in \cref{fig:heterogeneous} and they are reproduced for comparison purposes.}
    \label{fig:varyingMeanC}
\end{figure}

\begin{figure}[ht!]
    \centering
    \includegraphics[width=\linewidth]{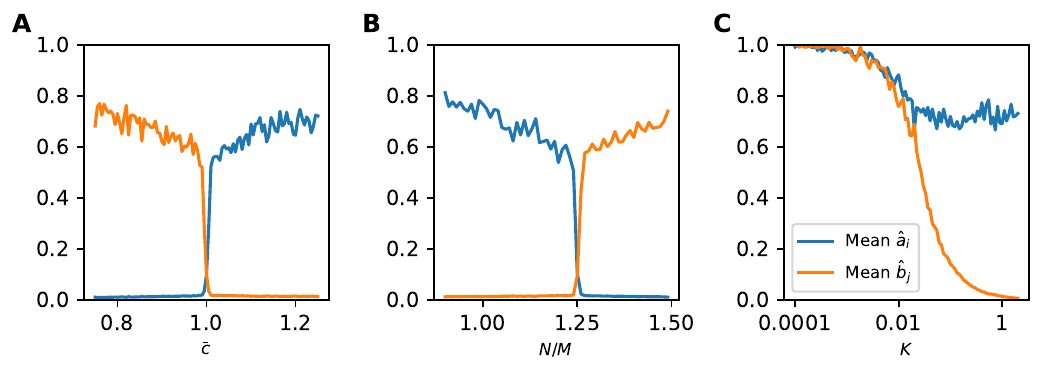}
    \caption{Mean acceptance probabilities at equilibrium of males and females ($\frac{1}{M}\cdot \sum_i\hat a_i$, $\frac{1}{N}\cdot \sum_i\hat b_j$) as we vary \textbf{(A)} mean target matching rate $\bar c$, \textbf{(B)} group size ratio $N/M$, \textbf{(C)} $K$. In (B), we keep $M=100$ fixed and vary $N$. Any parameters not explicitly varied are set as in \cref{fig:heterogeneous}.}
    \label{fig:varyParams}
\end{figure}
So far we have kept the sizes of the two groups of agents equal. Figures \labelcref{fig:varyParams}B and \labelcref{fig:varyingN} show what happens if we keep the distributions of the $c_i$'s and $d_j$'s constant (uniform with means $\bar c=1.25$ and $\bar d=1$) but vary the population ratio of males vs. females. We again see that equilibrium acceptance probabilities barely shift except when $N/M$ crosses the value $1.25$. Observe that this is equal to the ratio $\bar c/\bar d$, that is, the polarity changes abruptly when the group size ratio exceeds the (inverse) ratio of target matching rates.

\subsubsection{Selectivity of `in demand' group sensitive to rate of encounters.}

A particularly important parameter is the encounter rate, that is, the rate that agents meet others of the opposite group. The fact that this is a crucial parameter becomes evident once we note that if one is meeting people at a rate smaller than their target matching rate, then they would never be able to achieve their target, even if they matched with everyone they met. In the simulations so far, we have used $K=1$, which means that if there are $M=N=100$ agents of each type, then each individual agent is meeting $KM=KN=100$ agents of the other type per unit of time. This is much higher than the target matching rates, which we have drawn from distributions with means in the order of $\bar c\approx\bar d\approx 1$ match per unit of time. We would thus expect to see a qualitatively different outcome if $K$ was reduced by two orders of magnitude.

Figures \labelcref{fig:varyParams}C  and \labelcref{fig:varyingK} confirm this intuition. In fact, the acceptance probabilities of the `in demand' group (females in these simulations) rise as we decrease $K$, and they approach $1$ when $K<0.01$ (corresponding to each female meeting only $KN=1$ male per unit of time), while the acceptance probabilities of males increase to $1$ for $K<0.01$, but remain relatively unchanged otherwise. This underlines the fact that in order to observe the polarized distributions of selectivity, encounters need to occur frequently compared to the target matching rates. Reducing the encounter rate affects primarily the behavior of the `in demand' side.

\subsection{Extension with variable attractiveness}

In our base model, the acceptance probability of each individual did not depend on \textit{which} individual of the opposite type they encountered; as a result, there were no differences in the chances of different males or different females being accepted. In other words, all agents of each group were equally attractive. Here we introduce different levels of attractiveness.

A simple way to do this is to associate with each individual an attractiveness factor, $u_i$ for males and $v_j$ for females, taking values in $(0,1)$, which modulates the acceptance probability of the person by anyone of the opposite type. More precisely, we assume that the probability that male $i$ accepts female $j$ as a partner is $a_i(t)\cdot v_j$, and similarly the probability of female $j$ accepting male $i$ as a partner is $b_j(t)\cdot u_i$. Then, the rate of matching between the $i$-th male with any female becomes \begin{equation}
    \sum_jKa_i(t)v_jb_j(t)u_i=Ku_ia_i(t)\cdot \sum_jv_jb_j(t).
\end{equation}
Assuming again that individuals adapt their selectivity in order to reach their target matching rate, we may consider the system of equations

\begin{align}
\label{varyAttractivenessEq1a}
    \td{a_i(t)}{t} & =r \cdot \left[c_i-K\cdot u_ia_i(t)\cdot \sum_jv_jb_j(t)\right], \ \ \text{ subject to }a_i(t)\in [0,1],\\
\label{varyAttractivenessEq1b}
    \td{b_j(t)}{t} & =r \cdot \left[d_j-K\cdot v_jb_j(t)\cdot \sum_iu_ia_i(t)\right], \ \ \text{ subject to }b_j(t)\in [0,1].
\end{align}

If we make the change of variables $a_i'=u_ia_i$ and $b_j'=v_jb_j$, then the system becomes identical to \cref{mainEq1a,mainEq1b}, except that the allowed intervals for $a_i'$ and $b_j'$ are $[0,u_i]$ and $[0,v_j]$, respectively, instead of $[0,1]$. All convergence results continue to hold unchanged, except for the values of the equilibria.

\begin{figure}
    \centering
    \includegraphics[width=\linewidth]{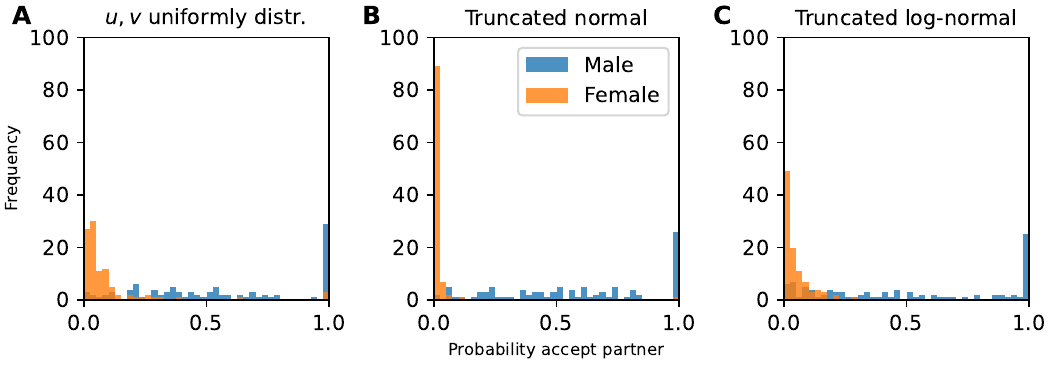}
    \caption{Equilibrium distributions from a simulation of the system of \cref{varyAttractivenessEq1a,varyAttractivenessEq1b}, with $u$ and $v$ drawn i.i.d. from \textbf{(A)} a uniform distribution on $(0,1)$. \textbf{(B)} /a normal distribution with $\mu=0.5$, $\sigma=0.25$, truncated to stay in the interval $(0.001,1)$, and \textbf{(C)} a log-normal distribution with $\mu=-2$, $\sigma=1$, capped at $1$. In each case, for the rate of encounters we set $K=(\bar u\bar v)^{-1}$, to account for the fact that successful encounters scale with both $\bar u$ and $\bar v$.}
    \label{fig:varyAttractiveness}
\end{figure}

\Cref{fig:varyAttractiveness} shows the equilibrium distributions when $u$ and $v$ are drawn from either a uniform distribution (\Cref{fig:varyAttractiveness}A), from a truncated normal (\Cref{fig:varyAttractiveness}B) or from a truncated log-normal (\Cref{fig:varyAttractiveness}C). A log-normal distribution creates the highest within-group inequality in attractiveness, because the attractiveness of the majority of individuals will be concentrated near low values, with a few being much more attractive. In all cases, we scale the rate of encounters $K$ by the inverse of the mean attractiveness of each group $\bar u$ and $\bar v$ to cancel out the effect of $u_i$ and $v_j$ taking values in $(0,1)$, which systematically decreases the chances of accepting someone as a partner.

As before, in all cases, the acceptance probability distributions at equilibrium differ markedly between the two groups. In the case of uniform attractiveness (\Cref{fig:varyAttractiveness}A), the difference is mildly less pronounced than before, with some females having less extremely low acceptance probabilities, although still quite low, with the exception of a small proportion of them. A smaller such effect is also visible in the case of the truncated log-normal (\Cref{fig:varyAttractiveness}C). Despite the high inequality in attractiveness in the case of the log-normal, the resulting distribution does not differ substantially from the case of constant attractiveness (\cref{fig:heterogeneous}).

\section{Proofs of mathematical results}
\label{sectionTechnical}

In this section we generalize the model described in \cref{sectionModel} and provide the proofs of the results of \cref{sectionAnalyticalResults} for this more general model. The reader who is not interested in the technical part may skip to the Discussion.

\subsection{Preliminaries}
\label{sectionPreliminaries}
To begin with, recall that in \cref{mainEq1a,mainEq1b} we said that the differential equations are subject to $a_i(t),b_j(t)\in [0,1]$. By this, we mean that if $a_i(t)$ (or $b_j(t)$) were to take a value above $1$, then we would set it to $1$. Observe that $a_i(t)$ can never fall below $0$, because $ \td{a_i(t)}{t}$ is always positive when $a_i(t)=0$. Therefore, \cref{mainEq1a,mainEq1b} can be more precisely written as (we suppress the dependence on $t$ for better readability) 

\begin{align}
\label{mainEq12a}
    \td{a_i}{t} & =\left\{\begin{array}{cc}
     r\cdot\left(c_i-K\cdot a_i\cdot \sum_jb_j\right), & \ \text{ if }a_i<1, \\
     \min\, \left\{r\cdot\left(c_i-K\cdot a_i\cdot \sum_jb_j\right),\ 0\right\}, & \ \text{ if }a_i=1.
    \end{array} \right.\\
\label{mainEq12b}
    \td{b_j}{t} & =\left\{\begin{array}{cc}
     r\cdot\left(d_j-K\cdot b_j\cdot \sum_ia_i\right), & \ \text{ if }b_j<1, \\
     \min\, \left\{r\cdot\left(d_j-K\cdot b_j\cdot \sum_ia_i\right),\ 0\right\}, & \ \text{ if }b_j=1.
    \end{array} \right.
\end{align}

In fact, we are going to prove our results for a more general system. Instead of assuming that agents adjust their acceptance probabilities \textit{linearly} with respect to how far their actual matching rate is from their target, we will allow them to use any rate of adjustment, subject to this adjustment always being in the direction of achieving their target rate. More precisely, we consider the system of equations
\begin{align}
\label{mainEq2a}
    \td{a_i}{t} & =\left\{\begin{array}{cc}
     g_i\left(K\cdot a_i\cdot \sum_jb_j\right), & \ \text{ if }a_i<1, \\
     \min\, \left\{g_i\left(K\cdot a_i\cdot \sum_jb_j\right),\ 0\right\}, & \ \text{ if }a_i=1.
    \end{array} \right.\\
\label{mainEq2b}
    \td{b_j}{t} & =\left\{\begin{array}{cc}
     h_j\left(K\cdot b_j\cdot \sum_ia_i\right), & \ \text{ if }b_j<1, \\
     \min\, \left\{h_j\left(K\cdot b_j\cdot \sum_ia_i\right),\ 0\right\}, & \ \text{ if }b_j=1,
    \end{array} \right.
\end{align}
where for all $i\in \{1,\ldots,M\}$ and $\forall j\in \{1,\ldots,N\}$ we have:
\begin{itemize}
    \item $g_i,h_j:[0,\infty)\to\R $ are Lipschitz continuous.
    \item $g_i$ is positive on $[0,c_i)$, negative on $(c_i,\infty)$, and $g_i(c_i)=0$, for some $c_i>0$.
    \item $h_j$ is positive on $[0,d_j)$, negative on $(d_j,\infty)$, and $h_j(d_j)=0$, for some $d_j>0$.
\end{itemize}
The system of \cref{mainEq12a,mainEq12b} is clearly a special case of \cref{mainEq2a,mainEq2b}, with $g_i(x)=r\cdot (c_i-x)$ and $h_j(x)=r\cdot (d_j-x)$.

An equivalent way to express the above system is as a \textit{projected dynamical system} on the unit cube $[0,1]^{M+N}$ \cite{nagurney1996projected}. Specifically, let us denote
\begin{equation}
\label{GENERAL-mainEqSingle}
    x_i=\left\{\begin{array}{cc}
         a_i, & i=1,\ldots ,M ,\\
         b_{i-M}, & i=M+1,\ldots ,M+N
    \end{array} \right.
\end{equation}
and
\begin{equation}
\label{GENERAL-defFi}
    f_i(\bm x)=\left\{\begin{array}{cc}
         g_i\left(K\cdot x_i\cdot \sumdu{j=M+1}{M+N}x_j\right), & i=1,\ldots, M, \\
         h_{i-M}\left(K\cdot x_{i}\cdot \sumdu{j=1}{M}x_j\right), & i=M+1,\ldots, M+N. 
    \end{array} \right.
\end{equation}
Then, \cref{mainEq2a,mainEq2b} can be more concisely written as
\begin{equation}
\label{GENERAL-eqDynamicsUniform}
    \td{x_i}{t}=\left\{\begin{array}{cc}
         f_i(\bm x), & \ \text{ if }x_i<1, \\
         \min\, \{f_i(\bm x),0\}, & \ \text{ if }x_i=1.
    \end{array} \right.
\end{equation}

Now let $P(x)$ denote the projection on $[0,1]^{M+N}$, that is,
\begin{equation}
    P(x)=\left\{\begin{array}{cc}
         x, & x\in [0,1]^{M+N},\\
         \underset{z}{\arg\min}\,||z-x||, & x\notin [0,1]^{M+N},
    \end{array}
    \right.
\end{equation}
where $||\cdot||$ denotes the standard Euclidean norm. Define also the projection of the vector $v\in \R^{M+N}$ at $x\in[0,1]^{M+N}$ (with respect to $[0,1]^{M+N}$) by 
\begin{equation}
    \Pi(x,v)=\limd{\delta\to 0}\frac{P(x+\delta v)-x}{\delta}.
\end{equation}
Intuitively, $\Pi (x,\cdot)$ takes a `velocity' vector centered at $x$ and modifies it minimally so that the trajectory stays within $[0,1]^{M+N}$. \Cref{GENERAL-eqDynamicsUniform} can be written as
\begin{equation}
\label{GENERAL-eqProjectedEquation}
    \td{x}{t}=\Pi({\bf x},{\bf f}({\bf x})).
\end{equation}
which is a projected differential equation (on $[0,1]^{M+N}$). Because of possible discontinuities of the vector field on the boundaries $x_i=1$, classical results for solutions do not apply. However, because $f$ is Lipschitz on $[0,1]^{M+N}$, standard results in projected dynamical systems theory guarantee that, for any initial condition, the system exhibits a forward unique \textit{Carath\'eodory} solution, i.e., a trajectory $\bm x(t)$ that is differentiable for almost every $t\geq 0$ and, for those $t$ for which ${\td{{\bf x}(t)}{t}}$ exists, it satisfies \cref{GENERAL-eqProjectedEquation} \cite[Theorem 2.5]{nagurney1996projected}. That is:

\begin{proposition}
\label{GENERAL-existenceTheorem}
For any initial condition ${\bf x_0}\in [0,1]^{M+N}$, the system defined by \cref{GENERAL-eqProjectedEquation} admits a unique Carath\'eodory solution ${\bf x}:[0,\infty)\to  [0,1]^{M+N}$ with ${\bf x}(0)={\bf x_0}$.
\end{proposition}

\subsection{General properties of the equilibrium}

Here we provide some basic properties of the equilibrium.

\begin{proposition}
\label{GENERAL-propositionPositiveEquilibrium}
If $\left(  \hat{\mathbf{a}}\mathbf{,}%
\hat{\mathbf{b}}\right)  $ is an equilibrium of the system of \cref{mainEq2a,mainEq2b}, then%
\[
\forall i\in \{1,\ldots,M\}:\hat{a}_{i}>0\qquad \text{ and}\qquad \text{
}\forall j\in \{1,\ldots,N\}:\hat{b}_{j}>0\text{.}%
\]

\end{proposition}

\begin{proof}
Take any $i\in \{1,\ldots,M\}$. If $\hat{a}_{i}=1$ we are done (for this $i$).
If $\hat{a}_{i}<1$ then
\[
0=\frac{da_{i}}{dt}=g_i\left(K\hat a_i\sum_j\hat b_j\right).
\]
If we had $\hat{a}_{i}=0$ then we would also have $g_i\left(0\right)=0$, which contradicts the assumption that $g_i(x)>0$ whenever $x<c_i$ (and recall that $c_i>0$). We conclude that $\hat{a}_{i}>0$ for every $i$. The proof for $\hat{b}_{j}>0$ is similar.
\end{proof}

\begin{proposition}
\label{GENERAL-propositionEquilibriumProperties}\normalfont$\left(  \hat{\mathbf{a}}\mathbf{,}
\hat{\mathbf{b}}\right)  $ is an equilibrium of the system of \cref{mainEq2a,mainEq2b} if and only if all the following hold
\begin{align}
\sum_{j}\hat{b}_{j}  &  >0\qquad \text{ and}\qquad \text{ }\forall
i:\hat{a}_{i}=\min \left\{  \frac{c_{i}}{K\sum_{j}\hat{b}_{j}},1\right\}
\label{GENERAL-eq0015}\\
\sum_{i}\hat{a}_{i}  &  >0\qquad \text{ and}\qquad \text{ }\forall
j:\hat{b}_{j}=\min \left\{  \frac{d_{j}}{K\sum_{i}\hat{a}_{i}},1\right\}
\label{GENERAL-eq0016}%
\end{align}

\end{proposition}

\begin{proof}
The proof is divided into two parts.

\bigskip

\begin{itemize}
    \item \noindent {``$\Rightarrow$''}: Suppose $\left(  \hat{\mathbf{a}%
    },\hat{\mathbf{b}}\right)  $ is an equilibrium, i.e., $\td{{\bf a}}{t}=\td{{\bf b}}{t}=\mathbf{0}$ at $\left(  \hat{\mathbf{a}%
    },\hat{\mathbf{b}}\right)  $. We will show that (\ref{GENERAL-eq0015}%
    )-(\ref{GENERAL-eq0016}) hold. The conditions $\sum_{i}\hat{a}_{i}>0$, $\sum
    _{j}\hat{b}_{j}>0$ follow from Proposition \ref{GENERAL-propositionPositiveEquilibrium}. Furthermore,
    from $\td{a_{i}}{t}=0$, one of the following two must hold:%
    \[
    \left[  \hat{a}_{i}<1\text{ and }g_i\left(K\hat{a}_{i}\sum_{j}\hat
    {b}_{j}\right)=0\right]  \text{ or }\left[  \hat{a}_{i}=1\text{ and }%
    g_i\left(K\hat{a}_{i}\sum_{j}\hat
    {b}_{j}\right)\geq 0\right]
    \]
    which by the assumptions on $g_i$ imply that (recall from \cref{GENERAL-propositionPositiveEquilibrium} that $\sum_j\hat b_j>0$)
    \begin{equation}
    \hat{a}_{i}=\frac{c_{i}}{K\sum_{j}\hat{b}_{j}}<1\text{ or }%
    1=\hat{a}_{i}\leq \frac{c_{i}}{K\sum_{j}\hat{b}_{j}}.
    \end{equation}
    From these two we get%
    \begin{equation}
        \hat{a}_{i}=\min \left\{  \frac{c_{i}}{K\sum_{j}\hat{b}_{j}},1\right\}.    
    \end{equation}
    
    Similarly we get
    \begin{equation}
    \hat{b}_{j}=\min \left\{  \frac{d_{j}}{K\sum_{i}\hat{a}_{i}},1\right\}
    \end{equation}
    
    \item \noindent ``$\Leftarrow$'': Suppose \cref{GENERAL-eq0015,GENERAL-eq0016}
    hold. To show that $\left(  \hat{\mathbf{a}},\hat{\mathbf{b}}\right)
    $ is an equilibrium, take any $i$. We have
    \begin{equation}
    \hat{a}_{i}=\min \left\{  \frac{c_{i}}{K\sum_j\hat{b}_{j}},1\right\}.
    \end{equation}
    If $\hat{a}_{i}=\frac{c_{i}}{K\sum_{j}\hat{b}%
    _{j}}<1$ then we have
    \begin{equation}
    \td{a_i}{t}=g_i\left\{K\hat a_i\sum_{j}\hat{b}_{j}\right\}=0
    \end{equation}
    If $\hat{a}_{i}=1$ then we have
    \begin{equation}
    g_i\left(K\hat a_i\sum_{j}\hat{b}_{j}\right)\geq 0\Rightarrow \min \left\{  0,g_i\left(K\hat a_i\sum_{j}\hat{b}_{j}\right)\right\}  =0\Rightarrow \td{a_{i}}{t}=0.
    \end{equation}
    Hence $\td{a_{i}}{t}=0$ for all $i\in \left \{  1,\ldots ,M\right \}  $;
    similarly we prove that $\td{b_{j}}{t}=0$ for all $j\in \left \{
    1,\ldots ,N\right \}  $. It follows that $\left(  \hat{\mathbf{a}}%
    ,\hat{\mathbf{b}}\right)  $ is an equilibrium.
\end{itemize}
\end{proof}

In what follows, it will be useful to introduce the following functions:
\begin{equation}
    F,F_0,G:(0,\infty)\to \R,\hspace{2em} H:(0,\infty)\to \R,
\end{equation}
where
\begin{equation}
\label{GENERAL-auxFunctions}
\begin{aligned}
    G(x)=\sum_i\min\left(\frac{c_i}{Kx}, 1\right), & \hspace{1em} H(x)=\sum_j\min\left(\frac{d_j}{Kx}, 1\right),\\
    F(x)=G(H(x)), & \hspace{1em} F_0(x)=\frac{F(x)}{x}.
\end{aligned}
\end{equation}

\begin{proposition}
\label{GENERAL-lemmaEquilibriumCondition}
    If $x^*\in (0,M]$ satisfies $F_0(x^*)=1$, the point defined by
    \begin{equation}
    \label{GENERAL-eq13}
        \hat a_{i}=\min\,\left\{\frac{c_i}{K\sum_j\hat b_j},1\right\},\ \text{ and } \ \hat b_j=\min\,\left\{\frac{d_j}{Kx^*},1\right\},
    \end{equation}
    is an equilibrium of the system of \cref{mainEq2a,mainEq2b} and $x^*=\sum_i\hat a_i$. 
    
    Conversely, any equilibrium $(\hat a_{1},\ldots,\hat a_{M},\hat b_{1},\ldots ,\hat b_{N})$ of \cref{mainEq2a,mainEq2b} satisfies $F_0\left(\sum_i\hat a_i\right)=1$.
\end{proposition}
\begin{proof}
    First note that $F_0(x^*)=1$ implies $F(x^*)=x^*$. Also, from \cref{GENERAL-eq13,GENERAL-auxFunctions} we have that $\sum_j\hat b_j=H(x^*)$ and
    \begin{equation}
        \sum_i\hat a_i=G\left(\sum_j\hat b_j\right)=G(H(x^*))=F(x^*)=x^*.
    \end{equation}
    Substituting $\sum_i\hat a_i$ for $x^*$ into \cref{GENERAL-eq13}, we obtain 
    \begin{equation}
    \label{GENERAL-eq14}
        \hat a_{i}=\min\,\left\{\frac{c_i}{K\sum_j\hat b_j},1\right\}>0\ \text{ and } \ \hat b_j=\min\,\left\{\frac{d_j}{K\sum_i\hat a_i},1\right\}>0.
    \end{equation}
    By \cref{GENERAL-propositionEquilibriumProperties}, $(\hat a_1,\ldots ,\hat a_M,\hat b_1,\ldots \hat b_N)$ is an equilibrium.

    For the converse, by \cref{GENERAL-propositionEquilibriumProperties} we have that $\sum _j\hat b_j=H\left(\sum _i\hat a_i\right)$ and
    \begin{equation}
        \sum_i\hat a_i=G\left(\sum _j\hat b_j\right)=G\left(H\left(\sum_i\hat a_i\right)\right)=F\left(\sum _i\hat a_i\right).
    \end{equation}
    Therefore, $F_0\left(\sum_i\hat a_i\right)=\frac{F\left(\sum_i\hat a_i\right)}{\sum_i\hat a_i}=\frac{\sum_i\hat a_i}{\sum_i\hat a_i}=1$.
\end{proof}

\subsection{Existence and uniqueness of equilibrium}

\begin{lemma}
\label{GENERAL-differentiabilityMin}
    Let $f(x)=\min\,\left\{g(x),1\right\}$, where $g:I\to \R $ is non-decreasing and $I\subset \R$ is an interval. Then, $f$ is differentiable wherever $g$ is, except perhaps at a single point.
\end{lemma}
\begin{proof}
    Let $u=\inf \{x\in I:g(x)\geq 1\}$. Then on $I\cap (-\infty ,u)$, $f=g$, and on $I\cap (u,\infty)$, $f$ is constant. Therefore, $f$ is differentiable wherever $g$ is, except perhaps at $u$.
\end{proof}

\begin{lemma}
\label{GENERAL-lemma_g0}
    \begin{enumerate}
        \item $F_0(x)$ is non-increasing on $(0,M]$.
        \item If $\sum_ic_i\neq \sum_jd_j$, then $F_0(x)=1$ has a unique solution on $(0,M]$.
    \end{enumerate}
\end{lemma}

\begin{proof}
    \begin{enumerate}
        \item By \cref{GENERAL-differentiabilityMin} $F_0$ is differentiable except perhaps at a finite number of points. We write
        \begin{equation}
        \label{GENERAL-eq7}
            F_0(x)=\sum_i\min\,\left\{c_ih(x),\frac{1}{x}\right\},
        \end{equation}
        where
        \begin{equation}
        \label{GENERAL-eq5}
            h(x)=\frac{1}{\sum_j\min\, \left\{d_j,Kx\right\}}.
        \end{equation}
        Let $x\in(0,M]$ be such that $F'_0(x)$ exists and let $A=\{j:d_j>Kx\}$. We have
        \begin{equation}
        \label{GENERAL-eqdummy11}
            \sum_j\min\, \left\{d_j,Kx\right\}=K\sumdu{j\in A}{}x+\sumdu{j\notin A}{}d_j.
        \end{equation}
        Therefore,
        \begin{equation}
            h'(x)=-\frac{K\cdot |A|}{\left(\sum_j\min\, \left\{d_j,Kx\right\}\right)^2},
        \end{equation}
        where $|\cdot |$ denotes cardinality. Now let $B=\{i:c_ih(x)<\frac{1}{x}\}$ and note that
        \begin{equation}
        \label{GENERAL-eqdummy12}
        \begin{aligned}
            F'_0(x) & =\sum_ic_ih'(x)-\sumdu{i\notin B}{}\frac{1}{x^2}\\
            & =-\sumdu{i\in B}{}Kc_i\cdot |A|\cdot \left(\sum_{j}\min\, \left\{d_j,Kx\right\}\right)^{-2}-\frac{|B^c|}{x^2}.
        \end{aligned}
        \end{equation}
        We conclude that $F'_0(x)\leq 0$. Since $F_0$ is continuous and $F'_0(x)\leq 0$ for all but finitely many points (on which $F'_0(x)$ doesn't exist), $F_0$ is non-increasing on $(0,M]$.
        \item We first show uniqueness. Given that $F_0$ is non-increasing, if $F_0(x_1)=F_0(x_2)=1$, then $F_0(x)=1$ for all $x\in [x_1,x_2]$, hence also $F'_0(x)=0$ on the interior $(x_1,x_2)$. Let $x\in (x_1,x_2)$ and use the definitions of $A$ and $B$ from above. By \cref{GENERAL-eqdummy12}, because both of the terms on the right-hand side are non-positive, $F_0'(x)=0$ implies $A=B^c=\emptyset$. By \cref{GENERAL-eqdummy11}, $A=\emptyset $ implies that $h(x)=\left(\sum_jd_j\right)^{-1}$. Also, by \cref{GENERAL-eq7}, $B^c=\emptyset$ implies that
        \begin{equation}
            F_0(x)=\sum_ic_ih(x)=\frac{\sum_ic_i}{\sum_jd_j},
        \end{equation}
        which is $\neq 1$ by assumption. We conclude that $F_0(x)$ cannot have more than one solution. To show that it has at least one, note that for sufficiently small $x$, we have $h(x)=\frac{1}{KNx}$, hence $\limd{x\to 0^+}F_0(x)=\infty $. Moreover, $F_0(M)=\sum_i\min\,\left\{c_ih(M),\frac{1}{M}\right\}\leq 1$. By the intermediate value theorem, $F_0$ takes the value $1$ at least once on $(0,M]$.
    \end{enumerate}
    
\end{proof}

\begin{corollary}
\label{GENERAL-corollaryUniqueEquilibrium}
    If $\sum_ic_i\neq \sum_jd_j$, then the system of \cref{mainEq2a,mainEq2b} has a unique equilibrium $({\bf \hat a},{\bf \hat b})$.
\end{corollary}
\begin{proof}
    Combine \cref{GENERAL-lemmaEquilibriumCondition} with the second part of \cref{GENERAL-lemma_g0}.
\end{proof}

\subsection{Global attractiveness}
\label{GENERAL-sectionStability}
We will now show global attractiveness of the equilibrium. We assume for the rest of our results that $\sum_ic_i\neq \sum_jd_j$, so that the uniqueness of the equilibrium is guaranteed.

\begin{assumption}[Standing assumption throughout \cref{GENERAL-sectionStability,GENERAL-sectionEquilibriumValue}]
    $\sum_ic_i\neq \sum_jd_j$
\end{assumption}

Let $A(t)=\sum_ia_i(t)$, $B(t)=\sum_jb_j(t)$, and define similarly $\hat A$, $\hat B$.
Define $A^+=\limsup A(t)$, $A^-=\liminf A(t)$, and similarly for $B^+$, $B^-$. Also recall the definitions
\begin{equation}
    G(x)=\sum_i\min\left(\frac{c_i}{Kx}, 1\right),\hspace{2em} H(x)=\sum_j\min\left(\frac{d_j}{Kx}, 1\right),\hspace{2em} F=G\circ H.
\end{equation}

\begin{lemma}
\label{GENERAL-lemmaLiminfInequalities}
    $B^-\geq H(A^+)$, $B^+\leq H(A^-)$, $A^-\geq G(B^+)$, and $A^+\leq G(B^-)$
\end{lemma}
\begin{proof}
    We will only show that $B^-\geq H(A^+)$, the rest of the relations are proved similarly. Let $\epsilon >0$. If $b_j(t)<\frac{d_j}{K(A^++\epsilon)}\Leftrightarrow K(A^++\epsilon)b_j(t)<d_j$, then we can find $t_0\geq 0$ such that for all $t>t_0$ we also have $KA(t)b_j(t)<d_j$. Therefore, from the assumptions on $h_j$ we get that for $t>t_0$,
    \begin{equation}
    \begin{aligned}
         b_j(t)<\frac{d_j}{K(A^++\epsilon)}\ \Rightarrow\ h_j\left(Kb_j(t)\cdot A(t)\right)>0.
    \end{aligned}
    \end{equation}
    By continuity of $h_j$, we can even find some $\delta>0$, such that if $t>t_0$ and $b_j(t)$ lies on the compact set $\left[0,\frac{d_j}{K(A^++\epsilon)}-\epsilon\right]$ (for sufficiently small $\epsilon>0$), we have
    \begin{equation}
        h_j\left(Kb_j(t)\cdot A(t)\right)>\delta .
    \end{equation}
    That is, for $t>t_0$,
    \begin{equation}
    \label{eq61}
        b_j(t)\in \left[0,\frac{d_j}{K(A^++\epsilon)}-\epsilon\right]\Rightarrow h_j\left(Kb_j(t)\cdot A(t)\right)>\delta.
    \end{equation}
    Recall that if $b_j(t)<1$, then $\td{b_j(t)}{t}=h_j\left(Kb_j(t)\cdot A(t)\right)$, therefore \cref{eq61} gives
    \begin{equation}
        b_j(t)<\min\,\left\{\frac{d_j}{K(A^++\epsilon)}-\epsilon ,1\right\}\Rightarrow \td{b_j(t)}{t}>\delta>0
    \end{equation}
    for all $t>t_0$. In other words, for sufficiently large $t$, the vector field points to the right and with magnitude at least $\delta$ whenever $b_j(t)<\min\,\left\{\frac{d_j}{K(A^++\epsilon)}-\epsilon ,1\right\}$. We conclude that $\liminf b_j(t)\geq \min\,\left\{\frac{d_j}{K(A^++\epsilon)}-\epsilon ,1\right\}$. Taking $\epsilon\to 0$ we get $\liminf b_j(t)\geq \min\,\left\{\frac{d_j}{KA^+},1\right\}$ and summing over $j$, we get $B^-\geq H(A^+)$.
\end{proof}

\begin{lemma}
\label{GENERAL-lemmaConvergenceA}
    For any initial condition, $\limd{t\to\infty }A(t)=\hat A$.
\end{lemma}
\begin{proof}
    From \cref{GENERAL-lemmaLiminfInequalities} we have $A^-\geq G(B^+)>0$, and because $G$ is non-increasing, $B^+\leq H(A^-)$ implies $G(B^+)\geq G(H(A^-))$. Therefore,
    \begin{equation}
        A^-\geq G(B^+)\geq G(H(A^-))=F(A^-),
    \end{equation}
    or equivalently $F_0(A^-)=\frac{F(A^-)}{A^-}\leq 1$. By \cref{GENERAL-lemmaEquilibriumCondition,GENERAL-lemma_g0}, $F_0(x)$ is non-increasing on $(0,M]$ and  $F_0(\hat A)=1$. Therefore, $A^-\geq \hat A$. We similarly get $A^+\leq \hat A$ and since by definition $A^-\leq A^+$, we conclude that $A^-=\hat A=A^+=\limd{t\to\infty }A(t)$.
\end{proof}

The last lemma shows that the sum of the acceptance probabilities, $A(t)=\sum_ia_i(t)$, converges. To get that the individual components also converge we need one more lemma.

\begin{proposition}
\label{GENERAL-propositionStability}
    The unique equilibrium of \cref{mainEq2a,mainEq2b} is globally attractive.
\end{proposition}

\begin{proof}
    Let $\epsilon\in (0,\hat A)$. By continuity of $h_j$ and the fact that $h_j(x)<0$ whenever $x>d_j$, we can find some $\delta >0$ such that
    \begin{equation}
    \label{eq64}
        x\in [d_j+\epsilon ,K(M+\epsilon)]\Rightarrow h_j(x)<-\delta .
    \end{equation}
    Let $t_0\geq 0$ be such that $\hat A-\epsilon<A(t)<\hat A+\epsilon$ for all $t>t_0$. Then,
    \begin{equation}
    \label{eq66}
        t>t_0\ \text{ and }\ Kb_j(t)(\hat A-\epsilon)>d_j+\epsilon \Rightarrow Kb_jA(t)>d_j+\epsilon. 
    \end{equation}
    Observe also that $\hat A=\sum_i\hat a_i\leq M$, hence
    \begin{equation}
        Kb_j(t)A(t)\leq Kb_j(t)(\hat A+\epsilon)\leq K(M+\epsilon).
    \end{equation}
    Combining this with \cref{eq66}, we get
    \begin{equation}
        t>t_0\text{ and }\ Kb_j(t)(\hat A-\epsilon)>d_j+\epsilon\hspace{1em}\Rightarrow\hspace{1em} Kb_j(t)A(t)\in [d_j+\epsilon, K(M+\epsilon)].
    \end{equation}
    Further combining with \cref{eq64}, we obtain
    \begin{equation}
    \label{eq67}
        t>t_0\ \text{ and }\ b_j(t)>\frac{d_j+\epsilon}{K(\hat A-\epsilon)}\hspace{1em}\Rightarrow\hspace{1em} h_j\left(Kb_jA(t)\right)<-\delta.
    \end{equation}
    Thus, from the definition of $\td{b_j(t)}{t}$ (\cref{mainEq2b}) we get
    \begin{equation}
    \label{eq68}
        t>t_0\ \text{ and }\ b_j(t)>\frac{d_j+\epsilon}{K(\hat A-\epsilon)}\hspace{1em}\Rightarrow\hspace{1em}\td{b_j(t)}{t}<-\delta.
    \end{equation}
    We conclude that $\limsup b_j(t)\leq \frac{d_j+\epsilon}{K(\hat A-\epsilon)}$ and because $b_j(t)\leq 1$ is always true, we even have
    \begin{equation}
        \limsup b_j(t)\leq \min\,\left\{\frac{d_j+\epsilon}{K(\hat A-\epsilon)},1\right\}.
    \end{equation}
    To get a lower bound, in exactly the same way we get the following analogue of \cref{{eq67}}:
    \begin{equation}
    \label{eq69}
        t>t_0\ \text{ and }\ b_j(t)<\frac{d_j-\epsilon}{K(\hat A+\epsilon)}\Rightarrow\ h_j\left(Kb_jA(t)\right)>\delta,
    \end{equation}
    possibly for a different $\delta>0$ and where now $\epsilon\in (0,d_j)$. By \cref{mainEq2b}, when $b_j(t)<1$, we have $\td{b_j(t)}{t}=h_j\left(Kb_jA(t)\right)$, therefore \cref{eq69} implies
    \begin{equation}
    \label{eq70}
        t>t_0\ \text{ and }\ b_j(t)<\min\,\left\{\frac{d_j-\epsilon}{K(\hat A+\epsilon)},1\right\}\hspace{1em}\Rightarrow\hspace{1em}\td{b_j(t)}{t}>\delta.
    \end{equation}
    We conclude that $\liminf b_j(t)\geq \min\,\left\{\frac{d_j-\epsilon}{K(\hat A+\epsilon)},1\right\}$. Therefore, we have shown
    \begin{equation}
        \min\,\left\{\frac{d_j-\epsilon}{K(\hat A+\epsilon)},1\right\}\leq \liminf b_j(t)\leq \limsup b_j(t)\leq \min\,\left\{\frac{d_j+\epsilon}{K(\hat A-\epsilon)},1\right\}
    \end{equation}
    Taking $\epsilon\to 0$, we obtain $\limd{t\to\infty}b_j(t)=\min\,\left\{\frac{d_j}{K\hat A},1\right\}$, which by \cref{GENERAL-propositionEquilibriumProperties} is equal to $\hat b_j$. Summing over $j$, we get $B(t)\to \hat B$, and using this we show in a similar way that $a_i(t)\to \hat a_i$.
\end{proof}

\begin{proof}[Proof of \cref{theoremMain}]
    Combine \cref{GENERAL-existenceTheorem,GENERAL-corollaryUniqueEquilibrium,GENERAL-propositionStability}.
\end{proof}

\subsection{Equilibrium value}
\label{GENERAL-sectionEquilibriumValue}
Here we give the proofs of \cref{propositionHomogeneousEq,propositionIstarJstarExpressions,propositionHeterogeneousEq} that give the value of the equilibrium. Note that while these propositions are stated in \cref{sectionAnalyticalResults} for the system of \cref{mainEq1a,mainEq1b}, we prove them for the more general equations \labelcref{mainEq2a,mainEq2b}.

\begin{proof}[Proof of \cref{propositionHomogeneousEq}]
    From \cref{GENERAL-propositionEquilibriumProperties} we have
    \begin{equation}
    \label{GENERAL-eq55}
        \hat a_i=\min\,\left\{\frac{c}{K\sum_j\hat b_j},1\right\},\ \  \hat b_j=\min\,\left\{\frac{d}{K\sum_i\hat a_i},1\right\},
    \end{equation}
    which implies $\hat a_1=\ldots =\hat a_M=:\hat a$ and similarly $\hat b_1=\ldots =\hat b_N=:\hat b$. Therefore, \cref{GENERAL-eq55} simplifies to
    \begin{equation}
    \label{GENERAL-eq57}
        \hat a=\min\,\left\{\frac{c}{KN\hat b},1\right\},\ \  \hat b=\min\,\left\{\frac{d}{KM\hat a},1\right\},
    \end{equation}
    The result will follow if we show that $\hat a=1$. Suppose on the contrary that $\hat a<1$. Then, by \cref{GENERAL-eq57},
    \begin{equation}
    \label{GENERAL-eq56}
        \hat a=\frac{c}{KN\hat b}.
    \end{equation}
    By \cref{mainEq2b}, at equilibrium we have
    \begin{equation}
        0=\td{b_j}{t}\leq h_j\left(K\hat bM\hat a\right)=r\cdot(d-K\hat bM\hat a),
    \end{equation}
    because $h_j(x)=r\cdot(d-x)$ by \cref{mainEq1b} and the assumption $d_j=d$ for all $j$. We conclude that $d\geq K\hat bM\hat a$. Substituting $\hat a$ from \cref{GENERAL-eq56}, we obtain
    \begin{equation}
    \begin{aligned}
        d\geq K\hat bM\cdot \frac{c}{KN\hat b}=\frac{Mc}{N},
    \end{aligned}
    \end{equation}
    which is impossible because by assumption $Mc>Nd$. We conclude that $\hat a=1$. Substituting this into the expression for $\hat b$ in \cref{GENERAL-eq57}, we obtain $\hat b=\min\,\left\{\frac{d}{KM},1\right\}$, which completes the proof.
\end{proof}

For \cref{propositionIstarJstarExpressions} we will need a couple of lemmas.

\begin{lemma}
\label{GENERAL-lemmaF0djExpression}
    Assume that $d_1\leq\ldots\leq d_N$ and let $x\in (0,M]$. We have
    \begin{align}
    \label{GENERAL-eq29}
        F_0(x) & =\sum_i\min\,\left\{\frac{c_i}{K(N-j')x+\sumdu{k=1}{j'}d_{k}},\frac{1}{x}\right\}\\
        & =\sum_i\min\,\left\{\frac{c_i}{K(N-j'')x+\sumdu{k=1}{j''}d_{k}},\frac{1}{x}\right\},
    \end{align}
    where
    \begin{equation}
        j'=\max\,\left\{j:\frac{d_j}{K}\leq x\right\},\ \ j''=\max\,\left\{j:\frac{d_j}{K}<x\right\}.
    \end{equation}
    with the convention that $\max\,\emptyset=0$. In particular, for any $j$,
    \begin{align}
    \label{GENERAL-eq31}
        F_0\left(\frac{d_j}{K}\right)=\sum_i\min\,\left\{\frac{c_i}{(N-j)d_j+\sumdu{k=1}{j}d_{k}},\frac{K}{d_j}\right\}.
    \end{align}
\end{lemma}
\begin{proof}
    From the definition of $j'$ and $j''$ and \cref{GENERAL-auxFunctions}, we have that
    \begin{equation}
    \begin{aligned}
    \label{GENERAL-eq60}
        H(x) & = \frac{\sumdu{k=1}{j'}d_k}{Kx}+(N-j')\\
        & = \frac{\sumdu{k=1}{j''}d_k}{Kx}+(N-j''),
    \end{aligned}
    \end{equation}
    because for terms $j''<j\leq j'$ we have $\frac{d_j}{Kx}=1$. Plugging the first of these expressions into the expression for $G$ in \cref{GENERAL-auxFunctions}, we get
    \begin{equation}
    \begin{aligned}
        F(x)=G(H(x)) & =\sum_i\min\, \left\{\frac{c_i}{KH(x)},1\right\}\\
        & = \sum_i\min\, \left\{\frac{c_ix}{\sumdu{k=1}{j'}d_k+K(N-j')x},1\right\}
    \end{aligned}
    \end{equation}
    and
    \begin{equation}
        F_0(x)=\frac{F(x)}{x}=\sum_i\min\, \left\{\frac{c_i}{\sumdu{k=1}{j'}d_k+K(N-j')x},\frac{1}{x}\right\}.
    \end{equation}
    The proof for $j''$ is exactly the same, except that we use the second expression from \cref{GENERAL-eq60}.
    
    If $x=\frac{d_j}{K}$ for some $j$, then $j'\geq j$ and $d_{j}=d_{j+1}=\ldots =d_{j'}$, therefore we have
    \begin{equation}
    \begin{aligned}
        (N-j)d_{j}+\sumdu{k=1}{j}d_k & =(N-j')d_j+(j'-j)d_j+\sumdu{k=1}{j}d_k\\
        & =(N-j')d_{j'}+\sumdu{k=1}{j'}d_k.
    \end{aligned}
    \end{equation}
    Thus, substituting $x=\frac{d_j}{K}$ into \cref{GENERAL-eq29}, we get
    \begin{equation}
    \begin{aligned}
        F_0\left(\frac{d_j}{K}\right) & =\sum_i\min\,\left\{\frac{c_i}{(N-j')d_{j'}+\sumdu{k=1}{j'}d_{k}},\frac{K}{d_{j'}}\right\}\\
        & =\sum_i\min\,\left\{\frac{c_i}{(N-j)d_j+\sumdu{k=1}{j}d_{k}},\frac{K}{d_j}\right\}.
    \end{aligned}
    \end{equation}
\end{proof}

\begin{lemma}
\label{GENERAL-lemmaSaturationCondition}
    Assume that $c_1\leq \ldots\leq c_M$ and $d_1\leq\ldots\leq d_N$. For any $i$, $\hat a_i<1$ if and only if
    \begin{equation}
    \label{GENERAL-eq58}
        \sum_j\min\,\left\{\frac{d_j}{(M-i)c_i+\sumdu{k=1}{i}c_{k}},\frac{K}{c_i}\right\}>1.
    \end{equation}
    Similarly, for any $j$, $\hat b_j<1$ if and only if
    \begin{equation}
    \label{GENERAL-eq54}
        \sum_i\min\,\left\{\frac{c_i}{(N-j)d_j+\sumdu{k=1}{j}d_{k}},\frac{K}{d_j}\right\}>1.
    \end{equation}
\end{lemma}
\begin{proof}
    We will only show the second statement; the first one follows by symmetry. By \cref{GENERAL-lemmaEquilibriumCondition}, we have $\sum_i\hat a_i=x^*$ where $x^*$ is the unique solution of $F_0(x)=1$, and by \cref{GENERAL-propositionEquilibriumProperties},
    \begin{equation}
        \hat{b}_{j}=\min \left\{  \frac{d_{j}}{K\sum_{i}\hat{a}_{i}},1\right\}=\min \left\{  \frac{d_{j}}{Kx^*},1\right\}.
    \end{equation}
    It follows that $\hat b_j<1$ if and only if $d_j/K<x^*$. Because $F_0$ is non-increasing, $F_0(x^*)=1$, and $F_0(x)\neq 1$ for any $x\neq x^*$ (\cref{GENERAL-lemma_g0}), we get
    \begin{equation}
        \label{GENERAL-eq47}
        \hat b_j<1\text{ if and only if }F_0(d_j/K)>1.
    \end{equation}
    The result now follows by substituting the expression for $F_0(d_j/K)$ found in \cref{GENERAL-lemmaF0djExpression} into \cref{GENERAL-eq47}.
\end{proof}

We are now ready to prove \cref{propositionIstarJstarExpressions,propositionHeterogeneousEq}
\begin{proof}[Proof of \cref{propositionIstarJstarExpressions}]
    Let $s_i$ denote the denominator of the first term inside the $\min\{\cdot,\cdot\}$ in \cref{GENERAL-eq58}, that is,
    \begin{equation}
        s_{i}=(M-i)c_i+\sumdu{k=1}{k}c_k.
    \end{equation}
    We show that $s_{i}$ is non-decreasing in $i$. We have
    \begin{equation}
    \begin{aligned}
        s_{i+1}-s_i & =(M-i-1)c_{i+1}+\sumdu{k=1}{i+1}c_k-(M-i)c_i+\sumdu{k=1}{i}c_k\\
        & =(M-i)(c_{i+1}-c_i)-c_{i+1}+c_{i+1}\\
        & =(M-i)(c_{i+1}-c_i)\geq 0,
    \end{aligned}
    \end{equation}
    because by assumption $c_1\leq\ldots\leq c_M$. Therefore, the left-hand side of \cref{GENERAL-eq58} is non-increasing in $i$. We conclude that
    \begin{equation}
    \begin{aligned}
        \hat i & =|\{i:\hat a_i<1\}\\
        & = \left|\left\{i:\sum_j\min\,\left\{\frac{d_j}{(M-i)c_i+\sumdu{k=1}{i}c_{k}},\frac{K}{c_i}\right\}>1\right\}\right|\\
        & =\max\,\left\{i:\sum_j\min\,\left\{\frac{d_j}{(M-i)c_i+\sumdu{k=1}{i}c_{k}},\frac{K}{c_i}\right\}>1\right\}.
    \end{aligned}
    \end{equation}
    where the first equality is the definition of $\hat i$ (\cref{defIstarJstar}) and the second line follows from \cref{GENERAL-lemmaSaturationCondition}. This concludes the proof of \cref{expressionIstar}. \Cref{expressionJstar} follows by symmetry.
\end{proof}

\begin{proof}[Proof of \cref{propositionHeterogeneousEq}]
    Because the $d_j$'s are in increasing order, we have from \cref{GENERAL-propositionEquilibriumProperties} that
    \begin{equation}
    \label{GENERAL-eq50}
        \hat j=\max\,\{j:\hat b_j<1\}=\max\,\left\{j:d_j/K<\sum_i\hat a_i\right\},
    \end{equation}
    which by \cref{GENERAL-lemmaEquilibriumCondition} becomes
    \begin{equation}
        \hat j=\max\,\{j:d_j/K<x^*\}.
    \end{equation}
    Thus, \cref{GENERAL-lemmaF0djExpression} gives
    \begin{equation}
    \label{GENERAL-eq49}
        F_0(x^*)=\sum_i\min\,\left\{\frac{c_i}{K(N-\hat j)x^*+\sumdu{j=1}{\hat j}d_{j}},\frac{1}{x^*}\right\}.
    \end{equation}
    The first expression inside the $\min\{\cdot,\cdot \}$ is smaller than the second if and only if
    \begin{equation}
    \label{GENERAL-eq48}
    \begin{aligned}
        c_ix^* & < K(N-\hat j)x^*+\sumdu{j=1}{\hat j}d_{j}\Leftrightarrow\\
        c_i & < K(N-\hat j)+\frac{\sumdu{j=1}{\hat j}d_{j}}{x^*}\Leftrightarrow\\
        c_i & < K\left[\sumdu{j=\hat j+1}{N}1+\sumdu{j=1}{\hat j}\frac{d_{j}}{K\sum_i\hat a_i}\right],
    \end{aligned}
    \end{equation}
    where in the last line we have used $x^*=\sum_i\hat a_i$ (\cref{GENERAL-lemmaEquilibriumCondition}). Now observe that from \cref{GENERAL-propositionEquilibriumProperties,GENERAL-eq50} and the fact that the $d_j$'s are in increasing order, we have that $\hat b_j=\frac{d_j}{K\sum_i\hat a_i}$ whenever $j\leq \hat j$ and $\hat b_j=1$ otherwise. Therefore, \cref{GENERAL-eq48} becomes
    \begin{equation}
        c_i <K\sum_j\hat b_j.
    \end{equation}
    By \cref{GENERAL-propositionEquilibriumProperties}, this is equivalent to $\hat a_i<1$, and from the definition of $\hat i$ (\cref{defIstarJstar}) and the fact that the $c_i$'s are in increasing  order, it is equivalent to $i\leq \hat i$. Therefore,
    \cref{GENERAL-eq49} becomes
    \begin{equation}
    \begin{aligned}
        F_0(x^*) & =\frac{\sumdu{i=1}{\hat i}c_i}{K(N-\hat j)x^*+\sumdu{j=1}{\hat j}d_j}+\frac{(M-\hat i)}{x^*}\\
        & =\frac{C^*}{K(N-\hat j)x^*+D^*}+\frac{(M-\hat i)}{x^*},
    \end{aligned}
    \end{equation}
    where we have used the definitions of $C^*$ and $D^*$ from \cref{defCDL}. By the definition of $x^*$, we have $F_0(x^*)=1$, therefore the last equation gives
    \begin{equation}
        1=\frac{C^*}{K(N-\hat j)x^*+D^*}+\frac{(M-\hat i)}{x^*}.
    \end{equation}
    simplifying we get
    \begin{equation}
    \begin{aligned}
        K(N-\hat j)(x^*)^2+[D^*-C^*-K(M-\hat i)(N-\hat j)]x^*-D^*(M-\hat i)=0.
    \end{aligned}
    \end{equation}
    If $\hat j=N$, we get $x^*=\frac{D^*(M-\hat i)}{D^*-C^*}$, otherwise
    \begin{equation}
        x^*=\frac{-L\pm \sqrt{L^2+4KD^*(M-\hat i)(N-\hat j)}}{2K(N-\hat j)},
    \end{equation}
    where $L=D^*-C^*-K(M-\hat i)(N-\hat j)$ (see \cref{defCDL}). The solution with the minus sign in front of the radical is always negative, therefore the only solution inside $(0,M]$ is the `+' one, which concludes the proof.
\end{proof}

\section{Discussion}

We have modeled the dynamics of selectivity of two groups of agents that seek to repeatedly form inter-group matches. At the core of our model was a feedback mechanism which is a consequence of the structure of such networks, namely the fact that an increase in selectivity of one group pushes the other group to become less selective, due to reduced matching opportunities, and vice versa. We showed using a combination of analytic techniques and numerical simulations that this feedback mechanism can act on (possibly tiny) differences in group size or composition to create extreme differences in selectivity between the two groups. This phenomenon can be observed even if the two groups' composition in terms of goals is almost (but not exactly) identical.

A striking implication of our results is that large behavioral differences need not imply equally large differences in underlying preferences. Moreover, individuals with relatively high target matching rates in the group that ends up being `in demand' may nevertheless become more selective than individuals with much lower target matching rates in the opposite group. In other words, observed behavior may align not with individual goals but rather with the opportunity structure generated collectively by the individuals of the two groups and the feedback between them. In repeated matching environments, group membership may therefore explain more variation in the observed selectivity than individual preferences do.

This point is particularly relevant for online or casual dating. Large observed differences in mate selection behavior between heterosexual men and women are often interpreted as reflecting underlying differences between the sexes, either in evolved mating psychology \cite{buss2019mate} or in social roles \cite{eagly1999origins}. Our results suggest that part of the observed behavioral asymmetry may instead arise from the dynamics of the matching process itself. Additionally, although the present model does not incorporate social learning or norm formation, the persistent behavioral asymmetries it predicts may contribute to the emergence and stabilization of gendered dating scripts \cite{rose1993young,eaton2011has}.

Beyond dating, our model applies more broadly to social settings where there are two types of agents that seek to match with one another and where matches are not exclusive. Examples include scientific papers and conference submissions, job interviews, apartment hunt, arranged marriage markets at the screening phase, various types of clubs and potential members, and other settings with similar structure. In all these cases, selectivity is determined not only by an individual's preferences, but it is also an adaptive response to the opportunities created by the collective behavior of the opposite side. For example, job candidates may have a target number of interviews they wish to obtain. Those receiving few interview invitations are likely to accept invitations even from less desirable employers, whereas those receiving many invitations can afford to be more selective. Employers face an analogous trade-off: when applications are scarce, they may invite weaker candidates to interview, whereas a large applicant pool allows them to become more selective. Our model predicts that, in sufficiently large job markets where there are many candidates and many employers that can potentially match, such feedback drives the market toward a highly asymmetric equilibrium in which one side finds it much easier to secure interviews than the other.

More generally, our model generates several empirical predictions applicable to any matching network that satisfies our assumptions. First, it predicts that selectivity polarization should become stronger when individuals encounter many potential partners relative to their desired matching rate, and weaker when encounters are scarce. Second, assuming that selectivity polarization is observed, the selectivities of new entrants of the two groups should be more similar across groups than selectivities of older members are. Third, the model predicts threshold-like transitions in selectivity when the total matching demand of one group crosses that of the other, whether through changes in population size or in average desired matching rates. Finally, interventions that alter the opportunity structure, e.g., by controlling encounter rates on a dating platform, may substantially change observed selectivity without changing users' underlying preferences.

Our work also contributes to the broader literature on matching by focusing on the dynamics through which selectivity emerges. Existing matching models typically derive selectivity by solving an optimization problem. Our model, instead, treats selectivity as a state variable that evolves through adaptive feedback. To make this adaptive feedback possible, agents have to keep track of their matching rate, which requires remembering at least the recent past.\footnote{Dependence on the past does not explicitly appear in our equations, because we take a mean-field approach and use the individual's expected value of matching rate in place of the actual (recently) experienced matching rate. The idea is that, as long as the adaptation is relatively slow with respect to the rate of encounters, the two rates will be approximately equal. A direction for future work would be to show this equivalence or study how the dynamics differ when one uses the actual experienced matching rates.} 

The model we present also complements evolutionary models of mate choice, parental care, and sexual selection \cite{kokko2002mutual,kokko2008parental,fromhage2016coevolution}, which formalize ideas of parental investment theory \cite{trivers2017parental}. Such models include an analogous feedback mechanism between males and females, but at an evolutionary timescale: under certain conditions, the more parental care one sex provides, the less parental care the other sex is predicted to evolve under natural selection, and likewise for choosiness (selectivity), so these traits evolve under natural selection over evolutionary timescales. In contrast, our model treats selectivity as a variable that individuals adjust repeatedly during their lifetime in response to matching opportunities, while holding their objectives fixed.

Several simplifying assumptions of the model should be kept in mind. We assumed that matches have zero duration, so that individuals immediately return to the pool of potential partners after matching, and that agents adapt their selectivity solely in response to their experienced matching rate. We also ignored strategic considerations and we only considered a single way that different levels of attractiveness may make preferences of different agents correlate. Relaxing these assumptions, for example by allowing temporary matches, strategic behavior, or by making agents progressively more reluctant to reduce their selectivity when it falls below certain levels, may produce qualitatively different dynamics and constitutes an interesting direction for future work. Yet the main conceptual contribution of the present paper remains valid, namely the fact that structural feedback mechanisms should be taken into account when evaluating behavioral asymmetries in repeated matching environments. Understanding when such asymmetries reflect genuine differences in preferences and when they instead arise from the dynamics of matching remains an important challenge for future theoretical and empirical work.

\section*{Acknowledgments}

The author thanks Pantelis Analytis, Athanasios Kechagias, Silvia Sonderegger, colleagues at the Institute for Advanced Study in Toulouse, members of the BirthRites group at the Max Planck Institute for Evolutionary Anthropology, members of the Collective Minds group at the Complexity Science Hub, participants of the 2025 French Regional Conference on Complex Systems, the 2025 Sunbelt conference, the ARS’25 Tenth International Workshop on Social Network Analysis, the 2026 European Human Behaviour and Evolution Association conference, and the 37th annual Human Behavior and Evolution Society conference for helpful comments and discussions.

The author was supported by ANR grant ANR-17-EURE-0010 (Investissements d’Avenir program) and by Sapere Aude grant 2065-00038B awarded to P. Analytis by the Independent Research Fund Denmark.

\section*{Code availability}

The Python code used for the simulations and figure generation is available on GitHub at \url{https://github.com/alexgelas/matchingDynamics}.

\bibliographystyle{apacite}
\bibliography{references}

@article{adachi2003search,
  title={A search model of two-sided matching under nontransferable utility},
  author={Adachi, Hiroyuki},
  journal={Journal of Economic Theory},
  volume={113},
  number={2},
  pages={182--198},
  year={2003},
  publisher={Elsevier}
}

@article{ashlagi2017unbalanced,
  title={Unbalanced random matching markets: The stark effect of competition},
  author={Ashlagi, Itai and Kanoria, Yash and Leshno, Jacob D},
  journal={Journal of Political Economy},
  volume={125},
  number={1},
  pages={69--98},
  year={2017},
  publisher={University of Chicago Press Chicago, IL}
}

@article{eaton2011has,
  title={Has dating become more egalitarian? A 35 year review using sex roles},
  author={Eaton, Asia Anna and Rose, Suzanna},
  journal={Sex roles},
  volume={64},
  number={11},
  pages={843--862},
  year={2011},
  publisher={Springer}
}

@article{rose1993young,
  title={Young singles' contemporary dating scripts},
  author={Rose, Suzanna and Frieze, Irene Hanson},
  journal={Sex Roles},
  volume={28},
  number={9},
  pages={499--509},
  year={1993},
  publisher={Springer}
}

@article{kokko2008parental,
  title={Parental investment, sexual selection and sex ratios},
  author={Kokko, Hanna and Jennions, Michael D},
  journal={Journal of evolutionary biology},
  volume={21},
  number={4},
  pages={919--948},
  year={2008},
  publisher={Blackwell Publishing Ltd Oxford, UK}
}

@article{buss2019mate,
  title={Mate preferences and their behavioral manifestations},
  author={Buss, David M and Schmitt, David P},
  journal={Annual review of psychology},
  volume={70},
  number={1},
  pages={77--110},
  year={2019},
  publisher={Annual Reviews}
}

@article{fromhage2016coevolution,
  title={Coevolution of parental investment and sexually selected traits drives sex-role divergence},
  author={Fromhage, Lutz and Jennions, Michael D},
  journal={Nature communications},
  volume={7},
  number={1},
  pages={12517},
  year={2016},
  publisher={Nature Publishing Group UK London}
}

@misc{SSRS2025OnlineDating,
  author       = {{SSRS}},
  title        = {The Public and Online Dating in 2025},
  year         = {2025},
  month        = feb,
  howpublished = {\url{https://ssrs.com/insights/the-public-and-online-dating-in-2025/}},
  note         = {Survey report, accessed 2026-06-25}
}

@article{lauermann2014stable,
  title={Stable marriages and search frictions},
  author={Lauermann, Stephan and N{\"o}ldeke, Georg},
  journal={Journal of Economic Theory},
  volume={151},
  pages={163--195},
  year={2014},
  publisher={Elsevier}
}

@book{nagurney1996projected,
  title={Projected dynamical systems and variational inequalities with applications},
  author={Nagurney, Anna and Zhang, Ding},
  volume={2},
  year={1996},
  publisher={Springer Science \& Business Media}
}

@incollection{lauermann2025matching,
  title={Matching with frictions},
  author={Lauermann, Stephan and N{\"o}ldeke, Georg},
  booktitle={Handbook of the Economics of Matching},
  volume={2},
  pages={579--642},
  year={2025},
  publisher={Elsevier}
}

@article{chade2017sorting,
  title={Sorting through search and matching models in economics},
  author={Chade, Hector and Eeckhout, Jan and Smith, Lones},
  journal={Journal of Economic Literature},
  volume={55},
  number={2},
  pages={493--544},
  year={2017},
  publisher={American Economic Association 2014 Broadway, Suite 305, Nashville, TN 37203-2425}
}

@article{kokko2002mutual,
  title={Why is mutual mate choice not the norm? Operational sex ratios, sex roles and the evolution of sexually dimorphic and monomorphic signalling},
  author={Kokko, Hanna and Johnstone, Rufus},
  journal={Philosophical Transactions of the Royal Society of London. Series B: Biological Sciences},
  volume={357},
  number={1419},
  pages={319--330},
  year={2002},
  publisher={The Royal Society}
}

@inproceedings{tyson2016first,
  title={A first look at user activity on tinder},
  author={Tyson, Gareth and Perta, Vasile C and Haddadi, Hamed and Seto, Michael C},
  booktitle={2016 IEEE/ACM International Conference on Advances in Social Networks Analysis and Mining ({ASONAM})},
  pages={461--466},
  year={2016},
  organization={IEEE}
}

@article{eagly1999origins,
  title={The origins of sex differences in human behavior: Evolved dispositions versus social roles.},
  author={Eagly, Alice H and Wood, Wendy},
  journal={American psychologist},
  volume={54},
  number={6},
  pages={408},
  year={1999},
  publisher={American Psychological Association}
}

@incollection{trivers2017parental,
  title={Parental investment and sexual selection},
  author={Trivers, Robert L},
  booktitle={Sexual selection and the descent of man},
  pages={136--179},
  year={1972},
  publisher={Routledge}
}

@article{gale1962college,
  title={College admissions and the stability of marriage},
  author={Gale, David and Shapley, Lloyd S},
  journal={The American mathematical monthly},
  volume={69},
  number={1},
  pages={9--15},
  year={1962},
  publisher={Taylor \& Francis}
}

@article{roth1992two,
  title={Two-sided matching},
  author={Roth, Alvin E and Sotomayor, Marilda},
  journal={Handbook of game theory with economic applications},
  volume={1},
  pages={485--541},
  year={1992},
  publisher={Elsevier}
}

\newpage

\appendix

\section{Supplementary figures}

\renewcommand\thefigure{\thesection.\arabic{figure}}    
\setcounter{figure}{0}    

\begin{figure}[ht!]
    \centering
    \includegraphics[width=1\linewidth]{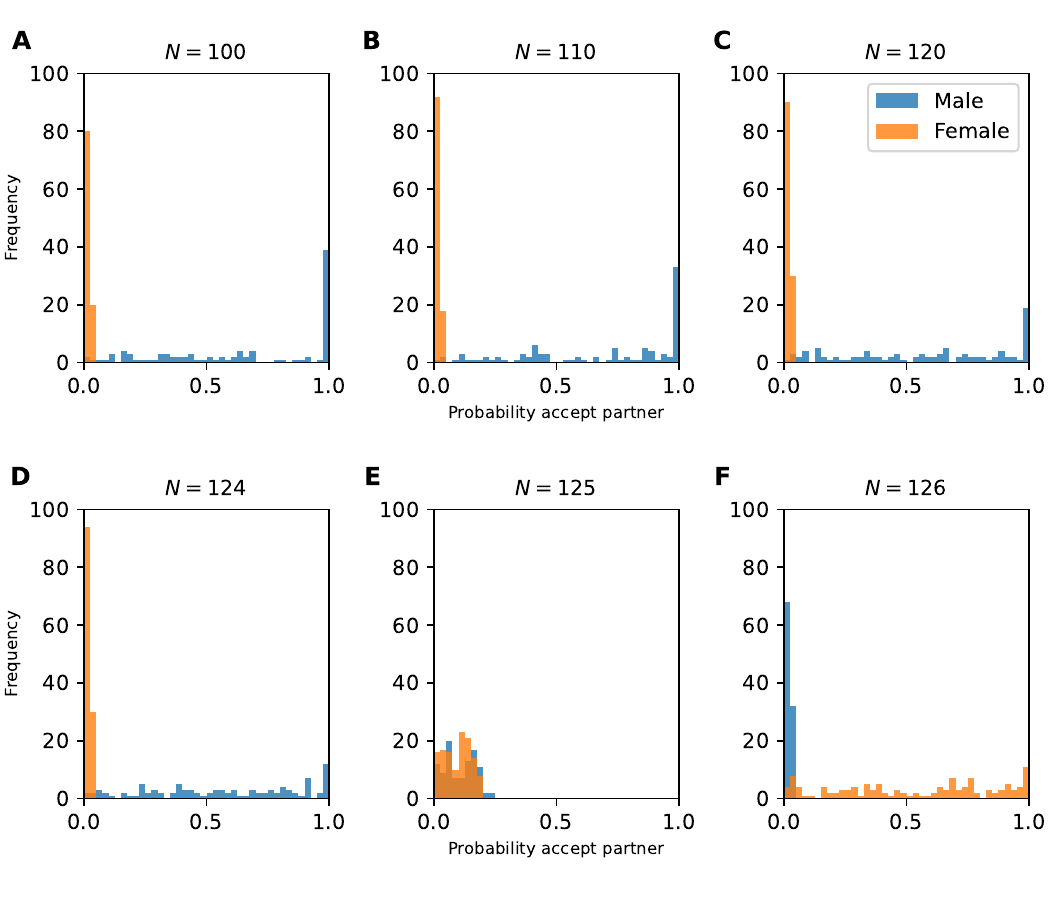}
    \caption{\textbf{Simulation of the system of \cref{mainEq1a,mainEq1b} with different population ratios.} \textbf{(A-F)} The number of females $N$ is indicated at the top of each panel. In all cases, the number of males is $M=100$. All other parameters are as in \cref{fig:heterogeneous}. The selectivity distributions at equilibrium remain essentially unchanged until $N/M\approx 1.25$ which is equal to the ratio of mean target matching rates ${\bar c}/{\bar d}$, and they reverse beyond that value. The parameter values in (A) are identical to those in \cref{fig:heterogeneous} and they are reproduced for comparison purposes.}
    \label{fig:varyingN}
\end{figure}

\begin{figure}[ht!]
    \centering
    \includegraphics[width=1\linewidth]{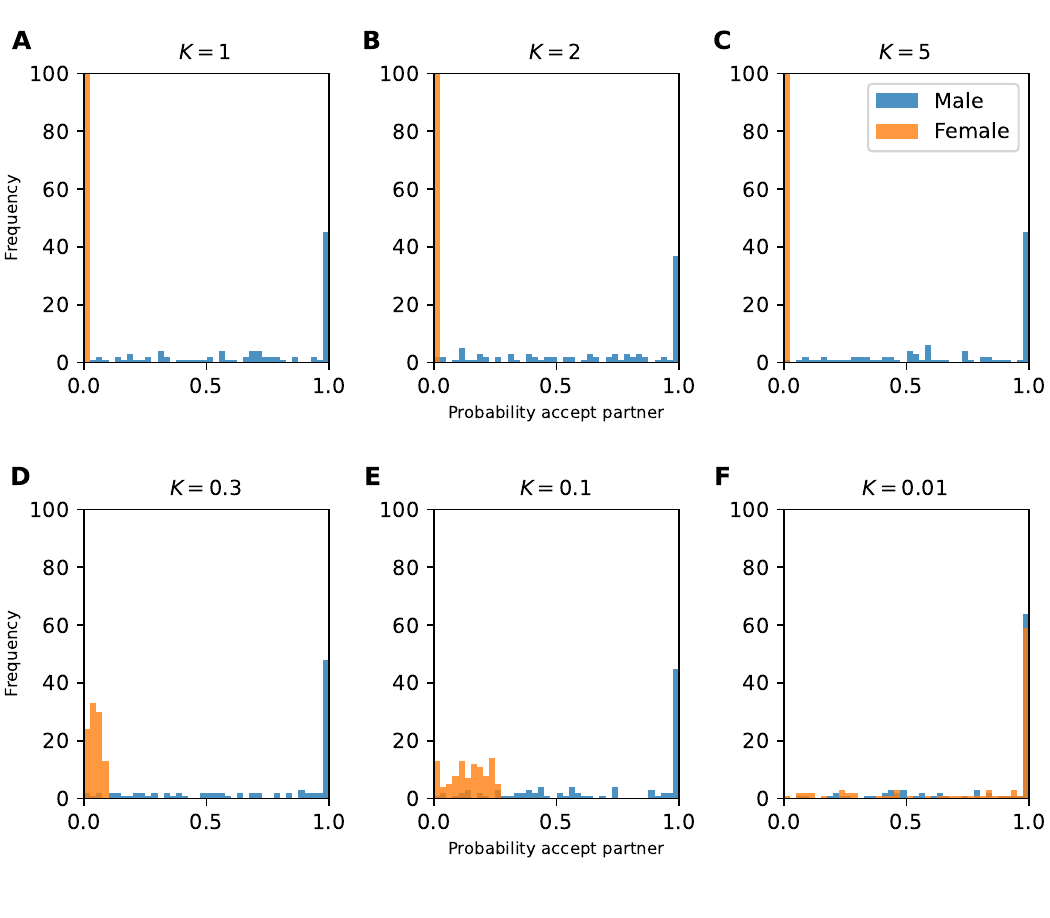}
    \caption{\textbf{Simulation of the system of \cref{mainEq1a,mainEq1b} with different encounter rates $K$.} \textbf{(A-F)} The rate of encounters (per agent of the opposite group) $K$ is indicated at the top of each panel. All other parameters are as in \cref{fig:heterogeneous}. Increasing $K$ has little effect on the equilibrium, except that females become proportionally more selective (not shown). Decreasing $K$ forces females to also become less selective and at $K=0.01$ many of them are non-selective. These values are meaningful only after taking into account the size of each group, which here are set to $M=N=100$. For example, $K=1$ and $M=100$ means that each male encounters $KM=100$ females per unit of time; recall that their target matching rate is on average $\bar d=1$ matches per unit of time. The fact that about half of them are non-selective for $K=0.01$ is thus to be expected. The parameter values in (A) are identical to those in \cref{fig:heterogeneous} and they are reproduced for comparison purposes.}
    \label{fig:varyingK}
\end{figure}

\end{document}